\documentclass[11pt,a4paper]{article}

\usepackage{amsmath,amsthm,amssymb,esint}
\usepackage{braket}
\usepackage{graphicx,color}
\usepackage[usenames,dvipsnames,svgnames,table]{xcolor}
\usepackage[colorlinks = true, linkcolor = blue, urlcolor  = blue, citecolor = violet]{hyperref}
\usepackage{tikz}
\usepackage[american]{circuitikz}
\usepackage[margin=1.25in]{geometry}

\newtheorem{corollary}{Corollary}[section]
\newtheorem{proposition}[corollary]{Proposition}
\newtheorem{lemma}[corollary]{Lemma}
\newtheorem{theorem}[corollary]{Theorem}
\newtheorem{remark}[corollary]{Remark}

\newcommand{\numset}[1]{\mathbf{#1}}
	\newcommand{\cc}{\numset{C}}
	\newcommand{\rr}{\numset{R}}
	
	\newcommand{\zz}{\numset{Z}}
	\newcommand{\nn}{\numset{N}}

	\newcommand{\one}{\bm{1}}
	
		\newcommand{\Exp}[1]{\mathrm{e}^{#1}}
	\newcommand{\ii}{\mathrm{i}}

	\makeatletter
	\providecommand*{\diff}%
		{\@ifnextchar^{\DIfF}{\DIfF^{}}}
	\def\DIfF^#1{%
		\mathop{\mathrm{\mathstrut d}}%
			\nolimits^{#1}\gobblespace}
	\def\gobblespace{%
		\futurelet\diffarg\opspace}
	\def\opspace{%
		\let\DiffSpace\!%
		\ifx\diffarg(%
			\let\DiffSpace\relax
		\else
			\ifx\diffarg[%
				\let\DiffSpace\relax
			\else
				\ifx\diffarg\{%
					\let\DiffSpace\relax
				\fi\fi\fi\DiffSpace}

	\renewcommand{\d}{\diff}

	\DeclareMathOperator{\tr}{tr}

	\renewcommand{\H}{\mathcal{H}}
	
			\newcommand{\cB}{\mathcal{B}}
	
	\newcommand{\CAR}{\operatorname{CAR}}

	\newcommand{\Ga}[1][]{\Gamma^{-}_{#1}}

	\newcommand{\dG}[1][]{\diff\Gamma_{#1}}

	\newcommand{\slim}{\operatorname{s-lim}\limits}
	\newcommand{\wlim}{\operatorname{w-lim}\limits}

\renewcommand{\one}{\mathbf{1}}
\newcommand{\tot}{_{\textnormal{tot}}}
\newcommand{\En}{_\textnormal{E}}
\newcommand{\Sa}{_\textnormal{S}}
\newcommand{\Bl}{_\textnormal{B}}
\newcommand{\BlSa}{_\textnormal{BS}}
\newcommand{\SaBl}{_\textnormal{SB}}
\newcommand{\EnSa}{_\textnormal{ES}}
\newcommand{\SaEn}{_\textnormal{SE}}
\newcommand{\Hb}{\mathcal{H}\Bl}

\newcommand{\Hs}{\mathcal{H}\Sa}
\newcommand{\Htot}{\mathcal{H}\tot}

\newcommand{\fU}{\mathfrak{U}}
\newcommand{\fY}{\mathfrak{Y}}
\newcommand{\T}{T}

\begin{document}

\title{Fermionic walkers driven out of equilibrium}
\date{}
\author{
	Simon Andr\'eys%
	\textsuperscript{1}
	\and
	Alain Joye%
	\textsuperscript{1}
  \and
  Renaud Raqu\'epas%
  \textsuperscript{1,2}
}

\maketitle
\begin{center}
\small
\begin{tabular}{c c c c}
   1.  Univ.\ Grenoble Alpes
	 	&&& 2. McGill University \\
		CNRS, Institut Fourier
	 	&&& Dept.\ of Mathematics and Statistics \\
	 	F-38\,000 Grenoble
	 	&&& 1005--805 rue Sherbrooke~Ouest \\
	 France
	 	&&& Montr\'eal (Qu\'ebec) ~H3A 0B9, Canada \\
		&&&
\end{tabular}
\end{center}

\begin{abstract}
  We consider a discrete-time non-Hamiltonian dynamics of a quantum system consisting of a finite sample locally coupled to several bi-infinite reservoirs of fermions with a translation symmetry. In this setup, we compute the asymptotic state, mean fluxes of fermions into the different reservoirs, as well as the mean entropy production rate of the dynamics. {Formulas are explicitly expanded to leading order in the strength of the coupling to the reservoirs.}
\end{abstract}

\section{Introduction}

\subsection{Motivation}

The mathematical description of the long time dynamics of many-body quantum systems coupled to several infinite reservoirs, and of the transport properties of non-equilibrium steady states they give rise to, is a long standing problem in quantum statistical mechanics, see {\it e.g.} \cite{AJPc}, \cite{JOPP11}. To achieve a better understanding of those important conceptual issues, many efforts have been devoted to the construction and analysis of models in various contexts or regimes. Following Jak\v{s}i\'{c} and Pillet \cite{JP01,JP02}, the main objectives for these models considered in the framework of open quantum systems are to establish the validity of the laws of thermodynamics, to derive the positivity of the entropy production rate and to analyse its fluctuations. It is desirable too to grasp model dependent salient features of the corresponding non-equilibrium steady states and currents they induce between the reservoirs.
See the following papers for a non exhaustive list of works dedicated to those questions in different contexts and regimes: \cite{Sp78,LS78,DdRM08,JPW14}, \cite{AJPP06,AJPP07,JOPP11} \cite{Ru00,Ru01,AP03,JLP13}, \cite{BJM06,BJM14,HJPR17,HJPR18,BJPP18,An20,BB20}, \cite{MMS07,MMS07b},...
In these works, the quantum dynamics of these systems derives from their Hamiltonians.

The last two decades have seen the emergence of a class of non-Hamiltonian models that proves efficient in modelling the quantum dynamics of complex systems, namely quantum walks.
A quantum walk (QW for short) arises as a unitary operator defined on a Hilbert space with basis elements associated to the vertices of an infinite graph, matrix elements coupling vertices of the graph a finite distance away from each other only. The QW discrete time dynamics implemented by iteration of the unitary operator has finite speed of propagation, and yields a dynamical system easily amenable to numerical investigation. By contrast to the models mentioned above, there is no Hamiltonian with natural physical meaning attached to a QW.
It was demonstrated over the years that QW provide useful approximations in various physical contexts and regimes,  see {\it e.g.} \cite{CC88, K-al09, Sp+13, Z-al10, M-al19, WaMa, TMT20}.
Furthermore, QW play an important role in quantum computing~\cite{AAKV01, KE03, Sa08, Por}, and they are also considered a quantum counterparts of classical random walks \cite{Gu08, Kon, APSS12}; see also the reviews \cite{Ve12,ABJ15}.

Given the versatility of QW and the wide range of physical situations they model and claims regarding different notions of quantum transport~\cite{KAG12,M+20}, it is natural to investigate their collective dynamical behaviour within the framework of open quantum systems when considered as indistinguishable quantum particles (quantum walkers) interacting with reservoirs. The first steps in this direction were performed in the work \cite{HJ17} and its generalisation \cite{Ra20}. They analyse the discrete time dynamics of an ensemble of fermionic QW on a finite sample, exchanging particles with an infinite reservoir of quasifree QW, and establish a form of return to equilibrium of the system. From a different perspective, these efforts can be viewed as an extension to discrete-time dynamics of a program which has mainly been carried out in Hamiltonian continuous-time settings.

Building up on \cite{HJ17, Ra20}, our aim is twofold. First we generalize the framework to the genuinely out of equilibrium situation in which the fermionic QW on the finite sample interact with several different quasifree QW reservoirs. Second, we analyse the onset of a non-equilibrium steady state in the sample and reservoirs, the development of related particle currents between the reservoirs, and establish strict positivity of the entropy production rate, in keeping with the program above. This closely parallels the work~\cite{AJPP07} on a Hamiltonian continuous-time model called the ``electronic black box''.

More precisely, each reservoir consists in noninteracting fermionic QW on a bi-infinite lattice, forced to hop to their left at discrete times. Hence the reservoirs free dynamics is the second quantization of a shift operator~$S$, while the free dynamics on the finite sample is the second quantization of an arbitrary one-particle unitary matrix~$W$. The interaction between the sample and each reservoir is given at the one-particle level by a unitary operator exchanging particles at specific sites of the sample and the reservoir, whose intensity is monitored by some coupling constant $\alpha$. The overall discrete dynamics is defined by  one step of interaction, one step of free evolution, one step of interaction, one step of free evolution and so on. Considering an initial state $\rho(0)$ given by a product of quasifree states in each reservoir defined by a translation invariant symbol $T$ (two-point function), and an arbitrary (even) state $\rho\Sa(0)$ in the sample, we determine the evolved state $\rho(t)$ for all time $t\in \nn$.

Under mild assumptions, we prove that $\rho(t)$ converges as $t\rightarrow\infty$ to a quasifree state, irrespective of the initial state in the sample, which allows us to determine the reduced asymptotic states in the sample and in the reservoirs. We extend the results of \cite{HJ17, Ra20} to our multi-reservoir setup by showing that the reduced asymptotic states in the sample is also a quasifree non-equilibrium state whose symbol~$\Delta^\infty$ is fully parametrized by $T$, $W$ and the coupling terms. Then, we turn to the flux into the different reservoirs and determine the steady state quantum mechanical expectation value of the flux observables, or QW currents.
We establish the validity of the first law of thermodynamics under very general conditions, and describe the conditions on the initial state $\rho(0)$ that induce nontrivial currents between the reservoirs. Assuming $\rho\Sa(0)$ is quasifree as well and considering the entropy production rate $\sigma(t)$ defined in terms the relative entropy between the symbols for the quasifree states at time~0 and~$t$, we prove that the asymptotic entropy production rate $\sigma^+=\lim_{t\rightarrow \infty}\sigma(t)$ exists and we characterize its strict positivity as a function of the initial state~$T$ of the reservoirs, the dynamics~$W$ in the sample and the couplings. Finally, we express the asymptotic entropy production rate~$\sigma^+$ in terms of the asymptotic currents between the reservoirs through the sample.

\subsection{Illustration}

For concreteness, let us illustrate our main results  in the case of an environment composed of two reservoirs. We consider that the Hilbert space of the environment is the fermionic second quantization of the space $\ell^2(\zz)\otimes \cc^2$ with a basis $\{\delta_l : l \in \zz \}$ of $\ell^2(\zz)$ and  a basis~$\{\psi_{\textnormal{L}},\psi_{\textnormal{R}}\}$ for~$\cc^2$. Heuristically
$\ell^2(\zz)\otimes \{\psi_{\textnormal{L}}\}$ supports the one-particle space a reservoir situated to the left of the sample and $\ell^2(\zz)\otimes \{\psi_{\textnormal{R}}\}$  the one-particle space a reservoir situated to the right of the sample. The Hilbert space of the sample is the fermionic second quantization of $\Hs$, a finite-dimensional space, so that the full one-particle space representing the sample and the environment is $\Htot= \ell^2(\zz) \otimes \cc^2 \oplus \Hs$. The free evolution of the sample is defined by a fixed one-particle unitary operator $W$ on $\Hs$, while that of the reservoirs is described by the one-particle shift operator on $\ell^2(\zz)$:
\[
	S \delta_l=\delta_{l-1}.
\]
To make the sample interact with the environment, we fix two orthonormal vectors $\phi_\textnormal{L}$ and $\phi_\textnormal{R}$ of $\Hs$, representing the position of the sample which are in contact respectively with the left and the right reservoir, and we suppose that walkers in the sample which are in the state $\phi_\textnormal{L}$ [resp. $\phi_\textnormal{R}$] can jump to the left reservoir [resp. the right reservoir], at the position indexed by zero in the environment.
For a given coupling strength $\alpha$, we describe the interaction by the one-particle unitary operator
\[\Exp{\ii\alpha ((\delta_0\otimes \psi_\textnormal{L}) \phi_\textnormal{L}^*+(\delta_0\otimes \psi_\textnormal{R}) \phi_\textnormal{R}^*+\mbox{\small h.c.})},
\]
where
\[
	(\delta_0\otimes \psi_\textnormal{L})\phi_\textnormal{L}^* : \eta \otimes\psi\oplus \varphi \mapsto  \langle \phi_\textnormal{L} , \varphi \rangle\,  \delta_0\otimes\psi_\textnormal{L}\oplus 0
\]
for all $ \eta \in \ell ^2(\zz), \psi\in \cc^2, \varphi\in \Hs$, and similarly for the index $R$. Here
``h.c.'' stands for hermitian conjugate, {\it i.e.} adjoint.
Eventually, each step of the overall evolution is represented by the fermionic second quantization of the unitary operator
\[
	\fU=\big( (S\otimes \one)\oplus W \big) \Exp{\ii\alpha ((\delta_0\otimes \psi_\textnormal{L}) \phi_\textnormal{L}^*+(\delta_0\otimes \psi_\textnormal{R}) \phi_\textnormal{R}^*+\mbox{\small h.c.})}~.
\]
Suppose that, at the level of Fock spaces, the left [resp.\ right] reservoir is initially a quasifree state with  translation invariant  symbol that has sufficiently regular Fourier transform~$f_\textnormal{L}$ [resp.~$f_\textnormal{R}$] defined on $[0, 2\pi]$ and the sample is initially in an arbitrary even state. Then, under some generic assumptions on $W$, the total system relaxes to a quasifree state whose zeroth order approximation in~$\alpha$ depends only on $W$, ~$f_\textnormal{L}$ and~$f_\textnormal{R}$ and not on the initial state on the sample. Moreover, a steady current of particles settles across the sample.
Assuming that $W$ has only simple eigenvalues ${\lambda_1}, ..., {\lambda_n}$ with {normalized} eigenvectors $\chi_1, ..., \chi_n$ we can express the current into the right reservoir in the limit $\alpha\rightarrow 0$ as
\[
  J_\textnormal{R} = \alpha^2 \sum_{i=1}^n \frac{|\braket{\chi_i, \phi_\textnormal{R}}|^2 |\braket{\chi_i, \phi_\textnormal{L}}|^2}{|\braket{\chi_i, \phi_\textnormal{R}}|^2+|\braket{\chi_i, \phi_\textnormal{L}}|^2} (f_\textnormal{L}({ -\ii}\log \lambda_i)-f_\textnormal{R}({ -\ii}\log \lambda_i))+O({ \alpha^4}),
\]
while the current $J_\textnormal{L}$ into the right reservoir is such that  $J_\textnormal{L}+J_\textnormal{R}=0$.
If $f_\textnormal{L}\left({\theta}\right)>f_\textnormal{R}\left({\theta}\right)$ for all $\theta\in \rr$ then the current is necessarily directed from the left to the right. However, if the function $f_\textnormal{R}$ and $f_\textnormal{L}$ cannot be compared on the unit circle, then we may choose the sign of the current $J_\textnormal{R}$ by tuning the eigenvalues of $W$.
This last property occurs when considering for example the one-particle free dynamics $W$ of a coined spin-$\tfrac 12$ quantum walk on the sample provided by a cycle with an even number~$n$ of vertices sketched in Figure~\ref{fig:ill}.

\begin{figure}
	\centering
\begingroup%
  \makeatletter%
  \providecommand\color[2][]{%
    \errmessage{(Inkscape) Color is used for the text in Inkscape, but the package 'color.sty' is not loaded}%
    \renewcommand\color[2][]{}%
  }%
  \providecommand\transparent[1]{%
    \errmessage{(Inkscape) Transparency is used (non-zero) for the text in Inkscape, but the package 'transparent.sty' is not loaded}%
    \renewcommand\transparent[1]{}%
  }%
  \providecommand\rotatebox[2]{#2}%
  \newcommand*\fsize{\dimexpr\f@size pt\relax}%
  \newcommand*\lineheight[1]{\fontsize{\fsize}{#1\fsize}\selectfont}%
  \ifx\svgwidth\undefined%
    \setlength{\unitlength}{235.35941034bp}%
    \ifx\svgscale\undefined%
      \relax%
    \else%
      \setlength{\unitlength}{\unitlength * \real{\svgscale}}%
    \fi%
  \else%
    \setlength{\unitlength}{\svgwidth}%
  \fi%
  \global\let\svgwidth\undefined%
  \global\let\svgscale\undefined%
  \makeatother%
  \begin{picture}(1,0.78981616)%
    \lineheight{1}%
    \setlength\tabcolsep{0pt}%
    \put(0,0){\includegraphics[width=\unitlength,page=1]{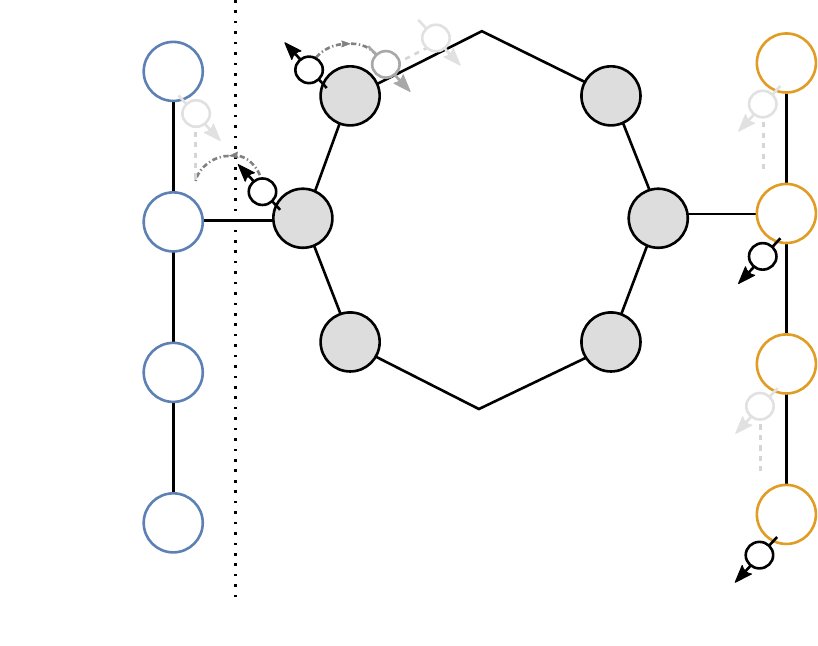}}%
    \put(0.00585741,0.50283931){\color[rgb]{0,0,0}\makebox(0,0)[lt]{\lineheight{0}\smash{\begin{tabular}[t]{l}$l=0$\end{tabular}}}}%
    \put(-0.00121711,0.68697965){\color[rgb]{0,0,0}\makebox(0,0)[lt]{\lineheight{0}\smash{\begin{tabular}[t]{l}$l=-1$\end{tabular}}}}%
    \put(0.01293198,0.31869915){\color[rgb]{0,0,0}\makebox(0,0)[lt]{\lineheight{0}\smash{\begin{tabular}[t]{l}$l=1$\end{tabular}}}}%
    \put(0.01293189,0.13455882){\color[rgb]{0,0,0}\makebox(0,0)[lt]{\lineheight{0}\smash{\begin{tabular}[t]{l}$l=2$\end{tabular}}}}%
    \put(0.5880572,0.50922697){\color[rgb]{0,0,0}\makebox(0,0)[t]{\lineheight{-0.001}\smash{\begin{tabular}[t]{c}$\mathcal{S}$\end{tabular}}}}%
    \put(0,0){\includegraphics[width=\unitlength,page=2]{fig-1.pdf}}%
    \put(0.28784159,0.00921961){\color[rgb]{0,0,0}\makebox(0,0)[t]{\lineheight{0}\smash{\begin{tabular}[t]{c}$J_\mathrm{L}$\end{tabular}}}}%
  \end{picture}%
\endgroup%

	\caption{The setup we are using to illustrate our results: walkers in a sample~$\mathcal{S}$ consisting of a cycle with~$8$ vertices can hop to and from two environments, one on the left and one on the right. Walkers at sites with a positive index~$l$ in the environment cannot have interacted with the sample yet.}
	\label{fig:ill}
\end{figure}

With a basis~$\{x_\nu \otimes e_{\tau} : \nu = 0, 1, \dotsc, n-1; \tau = -1,+1\}$ of~$\Hs = \ell^2(\{0,1,\dotsc,n-1\}) \otimes \cc^2$, an oft-studied model for the one-particle dynamics is given by the unitary
\[
  W := W_1 W_2
\]
where
\[
  W_1 := \sum_{\nu=0}^{n-1} \sum_{\tau = \pm 1} x_{\nu+\tau} \otimes e_{\tau} \braket{x_\nu \otimes e_\tau, \cdot\, }
\]
is a spin-dependent shift and
\[
  W_2 := \sum_{\nu=0}^{n-1} x_\nu x_\nu^* \otimes C_\nu
\]
encodes the rotation of a possibly position-dependent coin. In the special case where
\[
	C_\nu =
		\begin{pmatrix}
				\Exp{\ii \beta} \cos \varphi & \sin\varphi \\ -\sin\varphi & \Exp{-\ii\beta} \cos\varphi
		\end{pmatrix}
\]
for some real parameters $\beta,\varphi \in (0,\tfrac 12 \pi)$ independent of~$\nu$, the spectrum of~$W$ is easily shown to be contained in $\{\Exp{\ii u} : \varphi \leq \pm u \leq \pi-\varphi\}$ and is simple if $\beta \notin (2\pi/n) \zz$.

Before each step of the free walk, spin-up walkers located at sites~$0$ or~$\tfrac 12 n$ of the ring can be exchanged with those of the left or the right reservoirs. That means the interaction term above has
\[
	\phi_\textnormal{L}=x_0\otimes e_{+1}, \ \phi_\textnormal{R}=x_{n/2}\otimes e_{+1}.
\]

Moreover, the eigenvectors of $W$ being  explicitly computable, the current into the right reservoir eventually takes the form
\[
  J_\textnormal{R} = \alpha^2 \sum_{{\lambda_i} \in \operatorname{sp} W} \frac{\sin^2(2\varphi)|\sin\varphi - 1 + {\lambda_i}^2|^2}{4} (f_\textnormal{L}({ -\ii}\log \lambda_i)-f_\textnormal{R}({ -\ii}\log \lambda_i)) + O({ \alpha^4}).
\]

Therefore, if the parameter $\varphi$ is small and $f_\textnormal{L} > f_\textnormal{R}$ holds on open neighbourhoods of $\pi/2$ and $-\pi/2$ while $f_\textnormal{L} < f_\textnormal{R}$ on open neighbourhoods of $0$ and $\pi$, one gets that $J_\textnormal{R}>0$ for small couplings.
Considering $\ii W$ instead of $W$ for the same reservoirs yields $J_\textnormal{R}<0$ for small couplings; see Figure~\ref{fig:sp}. {The change of sign can equivalently be obtained by adding an appropriate common phase to the free dynamics of each reservoir\,---\,which is analogous to the change of sign that can occur by shifting the chemical potential in Hamiltonian systems for which the Landauer--B\"uttiker formula is valid.}

\begin{figure}
	\centering
	\hspace{.5in}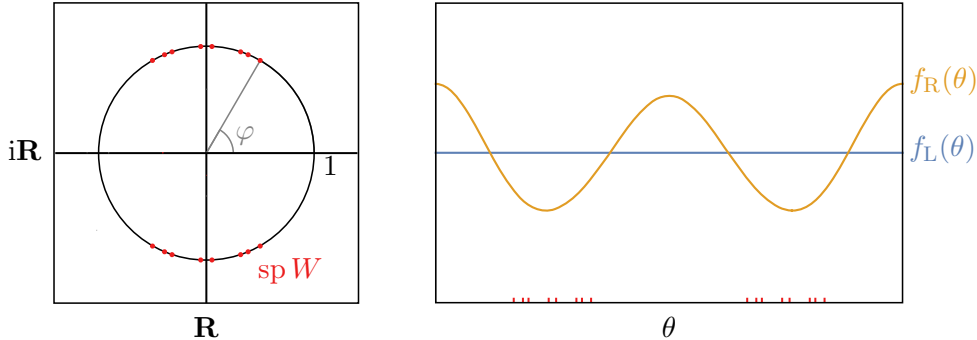
	\caption{Still in the setup of Figure~\ref{fig:ill}, with $\varphi = \tfrac 13 \pi$ and $\beta = 0.1$, the spectrum of~$W$ (on the left, in red online) lies in closed cones of opening $\pi-2\varphi$ about the imaginary axis. The corresponding arguments are values of~$\theta$ for which $f_\textnormal{L}(\theta) > f_\textnormal{R}(\theta)$ (on the right). Multiplying $W$ by a phase of~$\ii$  amounts to a rotation by quarter turn of the spectrum on the left and to a horizontal shift for the arguments on the right, leading to the opposite inequality.}
	\label{fig:sp}
\end{figure}

\subsection{Structure of the paper}

The paper is organized as follows: The next section is devoted to the description of our quantum dynamical system in a fairly general abstract framework. The long time asymptotic state is determined in Section 3, together with its restrictions to the sample and the reservoirs. Section 4 analyses the properties of the steady state currents of particles across the sample, while the study of the entropy production rate is conducted in Section 5. Eventually, the small coupling regime is analyzed in Section 6, and the paper closes with the proofs of certain results.

\paragraph*{Acknowledgements} The research of the authors is partially supported by the French National Agency through the grant NonStops (ANR-17-CE40-0006). The research of S.\,A. is supported by the French National Research Agency in the framework of the ``Investissements d’avenir'' program (ANR-15-IDEX-02). The research of R.\,R. is partially funded by the National Sciences and Engineering Research Council of Canada. The authors would like to thank the anonymous referees for their comments, which improved the quality of the presentation.

\section{The setup}

\subsection{The spaces and one-particle dynamics}
\label{ssec:spaces}

Let~$\Hs$ be a finite-dimensional Hilbert space. Throughout the paper, our terminology implicitly relies on the assumption that~$\Hs$ is the appropriate Hilbert space for the description of a quantum walker on a finite graph, sometimes referred to as a \textit{sample}. An evolution for a quantum walker on a slight extension of this sample could be encoded in a unitary operator~$Z$ on a Hilbert space of the form~$\Hb \oplus \Hs$ where~$\Hb$ is the Hilbert space associated the extension. With respect to this direct sum decomposition, the blocks of~$Z$, say
\begin{equation}
\label{eq:def-Z}
  Z =
  \begin{pmatrix}
    C & Z\BlSa \\ Z\SaBl & M
    \end{pmatrix},
\end{equation}
should satisfy
\begin{equation}\label{eq:zunitary}
  \left\{
\begin{array}{c c}
  C^*C+Z\SaBl^*Z\SaBl=\one, &
  C^* Z\BlSa+{Z\SaBl^*}M=0, \\
  Z\BlSa^*C+M^* Z\SaBl=0, &
  Z\BlSa^*Z\BlSa+M^*M=\one,
\end{array}
  \right.
\end{equation}
for the identity~$Z^*Z = \one$ to hold (and similarly for $ZZ^* = \one$). The off-diagonal blocks~$Z\BlSa$ and~$Z\SaBl$ describe the coupling between the sample and its extension and the bock~$M$ is thought of as an effective perturbation of a unitary~$W$ on~$\Hs$.

The Hilbert space
\[
  \Htot := (\ell^2(\zz) \otimes \Hb) \oplus \Hs
\]
for some finite-dimensional Hilbert space~$\Hb$ is instead suitable for the description of situations where the sample is interacting with an infinite environment which has a certain translation-invariant structure. Let us construct a single-particle unitary operator~$\fU$ on~$\Htot$ such that powers of~$\fU$ can be interpreted as successive interactions of the type encoded in~$Z$ with different blocks of this infinite environment.

Let~$(\delta_l)_{l \in \zz}$ be the canonical basis of~$\ell^2(\zz)$ and let
\begin{align*}
  S :   \ell^2(\zz)   &\to \ell^2(\zz) \\
        \delta_l   &\mapsto \delta_{l-1}
\end{align*}
be the shift operator and~$U : \Hb \to \Hb$ be an arbitrary unitary operator.
We set
\begin{equation}
  \label{eq:def-fU}
  \fU :=
  \begin{pmatrix}
    (S \otimes U)(P_0^\perp \otimes \one + P_0 \otimes C)   & S \delta_0 \otimes UZ\BlSa \\
    \delta_0^* \otimes Z\SaBl                      & M
    \end{pmatrix},
\end{equation}
on~$\Htot$ where $P_0 : \ell^2(\zz) \to \ell^2(\zz)$ is the orthogonal projector on the span of~$\delta_0$ and~$P_0^\perp := \one - { P_0}$.
Here, $\delta_0 \in \ell^2(\zz)$ is identified with a linear operator from~$\cc$ to~$\ell^2(\zz)$, so that e.g.\ $\delta_0 \otimes Z\BlSa$ can indeed be considered as an operator from~$\Hs \simeq \cc\otimes \Hs$ to~$\ell^2(\zz)\otimes\Hb$.
The unitary operator $\fU$ is quite natural to consider: it acts as the unitary operator~$Z$ on the space $\{\delta_0\}\otimes \Hb\oplus \Hs \simeq \Hb\oplus \Hs$ and then as the free evolution $S\otimes U$ on~$\ell^2(\zz)\otimes \Hb$; see Section~\ref{sec:small} for the discussion of the explicit link with the Introduction.

We make the following assumptions on the effective dynamics in the sample which was previously discussed in~\cite{HJ17,Ra20} in important examples.
\begin{description}
  \item[Assumption (Sp)] The spectrum of~$M$ is contained in the interior of the unit disk.
\end{description}

\subsection{The initial state in Fock space}
\label{ssec:init-state}

To describe the evolution of a varying number of fermionic walkers in the system we consider observables in the canonical anticommutation algebra~$\CAR(\Htot)$ represented on the fermionic Fock space~$\Ga(\Htot)$.

The fermionic Fock space space~$\Ga(\Htot)$ is unitarily equivalent to the tensor product~$\Ga(\ell^2(\zz) \otimes \Hb) \otimes \Ga(\Hs)$ of Fock spaces through a map~$\mathsf{E}$ such that
\[
  \mathsf{E} a^*(v \oplus w) \mathsf{E}^{-1} = a^*(v) \otimes \one + (-1)^{\dG(\one)} \otimes a^*(w)
\]
for all~$v \in \ell^2(\zz) \otimes \Hb$ and $w \in \Hs$.
This map associates quasifree states on~$\CAR(\Htot)$ with a symbol of the form $\T \oplus \Delta$ for some suitable~$\T : \ell^2(\zz) \otimes \Hb \to \ell^2(\zz) \otimes \Hb$ and~$\Delta : \Hs \to \Hs$
with the product of the corresponding quasifree states on $\CAR(\ell^2(\zz) \otimes \Hb)$ and~$\CAR(\Hs)$ respectively. We refer the reader to~\cite[\S{5.1,6.3}]{AJPP06} for a more thorough discussion.

We recall that $\omega_{\T}$ is a gauge-invariant quasifree state on~$\CAR(\ell^2(\zz) \otimes \Hb)$ with symbol~$0 \leq T \leq \one$ if
\[
  \omega_{\T}\big[a^*(v_n) \dotsb a^*(v_1)a(u_1) \dotsb a(u_m)\big] = \delta_{n,m}\det [\braket{u_i,\T v_j}]
\]
for all choices of~$v_1, \dotsc, v_n, u_1, \dotsc, v_m \in \ell^2(\zz) \otimes \Hb$, where~$a^*$ and~$a$ are the usual Fock space creation and annihilation operators\,---\,and similarly for other spaces. We refer the reader to~\cite{DFP08} for the basic theory of such states.

We will always make either of the following two assumptions on the initial state of the system, the second being technically more convenient and allowing simpler expressions for quantities of interest:
\begin{description}
  \item[Assumption (IC)] The initial state of the joint system is of the form
  \[
    \rho(0) = \mathsf{E}^{-1}(\omega_{\T} \otimes \rho\Sa)\mathsf{E}
  \]
  where $\rho_{\textnormal{S}}$ is an even state on the algebra~$\CAR(\Hs)$
  and~$\omega_{\T}$ is a gauge-invariant quasifree state on the algebra~$\CAR(\ell^2(\zz)\otimes\Hb)$ with symbol $T : \ell^2(\zz)~\otimes~\Hb \to \ell^2(\zz)~\otimes~\Hb$, $0 \leq~\T~\leq~\one$ such that
  \[
    [\T, S\otimes U] = 0.
  \]
  In addition, we assume that
  \[
    \sum_{l\in\zz} |l| \|(\delta_0^* \otimes \one) \T (\delta_l \otimes \one)\|< \infty.
  \]
  \item[Assumption (IC+)] The initial state of the joint system is as in~(IC) with $\rho\Sa$ also quasifree, with a symbol~$\Delta : \Hs \to \Hs$; equivalently, the initial state is a quasifree state with a density of the form~$\T \oplus \Delta$. Moreover, it is bounded away from~$0$ and~$\one$ in the sense that there exists~$\epsilon > 0$ such that~$\epsilon\one \leq \T \leq (1-\epsilon)\one$.
\end{description}
We also suppose that
\begin{description}
  \item[Assumption (Bl)] There exists a family~$\{\Pi_k\}_{k=1}^{n\Bl}$ of orthogonal projections summing to the identity on~$\Hb$ such that
  \[
    [U,\Pi_k] = 0,
  \]
  and
  \[
    [\T, \one \otimes \Pi_k] = 0
  \]
  for each~$k = 1,\dotsc, n\Bl$.
\end{description}
Note that~Assumption (Bl) technically always holds with~$n\Bl = 1$ and~$\Pi_1 = \one$, but is thought of as a separation of the environment into~$n\Bl$ different bi-infinite reservoirs of fermions,  with their own dynamics, which only interact through the sample. {Also, the case with $\operatorname{rank} \Pi_k = 1$ for each~$k$ will allow more explicit computations of some important quantities.}

In terms of the linear operators
\begin{equation}
  \label{eq:pre-blocks-T}
  \T_{n,m} := (\delta_n^* \otimes \one )\T(\delta_m \otimes \one)
\end{equation}
on~$\Hb$, referred to as \emph{blocks}, the commutation assumption in~(IC) becomes the requirement that
\begin{equation}
  \label{eq:blocks-T}
  \T_{n,m} = U^{-n} \T_{0,m-n} U^n.
\end{equation}
for all~$n, m \in \zz$.

\subsection{Relation to repeated interaction systems}

To clarify the place of our model in the zoo of discrete-time quantum dynamics, we comment on its relation to \emph{repeated interaction systems} (\textsc{ris}). This subsection can be skipped on a first reading. Consider the effective one-step dynamics in the sample
\[
	\Lambda_1(\rho) := \tr_{\Ga(\ell^2(\zz) \otimes \Hb)} [\Gamma(\fU^*)(\omega_{\T} \otimes \rho)\Gamma(\fU)],
\]
starting with an initial state as in Assumption~(IC+). A straightforward computation  making use of the Bogolyubov relation shows that~$\Lambda_1(\rho)$ is a quasifree state with symbol
\[
	\Delta^1 = M \Delta M^* + Z\SaBl \T_{0,0} Z\SaBl^*.
\]
Repeatedly applying the map~$\Lambda_1$, say~$t$ times to obtain a quasifree state with symbol
\[
	\Delta^t_{\textnormal{RIS}} = M^t \Delta (M^*)^t + \sum_{m=0}^{t-1} M^m Z\SaBl  \T_{0,0} Z\SaBl^* (M^*)^m,
\]
is an instance of a \textsc{ris}, as noted in the single reservoir setups of~\cite{HJ17, Ra20}. One can show that this \textsc{ris} picture coincides precisely with what happens at the level of the sample in the setup of Subsections~\ref{ssec:spaces} and~\ref{ssec:init-state} if $T_{n,m} = 0$ whenever $n \neq m$.
For example, compare our setup with~$Z = \exp[-\ii \tau (k\En \oplus k\Sa  + \lambda v)]$ for some one-particle selfadjoints operators~$k\En, k\Sa$ and~$v$ and compare the resulting dynamics on Fock space to the content of Section~II of~\cite{BJM14} using the exponential law for fermions.

However, in general, the effective dynamics in the sample
\[
	\Lambda_t(\rho) := \tr_{\Ga(\ell^2(\zz) \otimes \Hb)} [\Gamma(\fU^*)^t(\omega_{\T} \otimes \rho)\Gamma(\fU)^t]
\]
\emph{need not} enjoy the semigroup property~$\Lambda_{t+t'} = \Lambda_t \circ \Lambda_{t'}$. Indeed, we will see in Remark~\ref{rem33bis} below that, under Assumption~(IC+), $\Lambda_t(\rho)$ is a quasifree state with density
\[
	\Delta^t= M^t \Delta (M^*)^t + \sum_{m=0}^{t-1}\sum_{n=0}^{t-1}  M^m Z \SaBl \T_{0,m-n} U^{n-m} Z\SaBl^* (M^*)^n.
\]
The difference between~$\Delta^t_{\textnormal{RIS}}$ obtained in the~\textsc{ris} scenario and our general $\Delta^t$ amounts to the terms with $n \neq m$ in the latter, which generically do not cancel out. More generally, tracing out at steps that are multiples of a number~$\tau \geq 2$ for which $\T_{0,m} = 0$ for $m > \tau$, a similar computation shows that the dynamics differs from the original one by terms with no particular structure for cancellation.

On the other hand, the fact that we obtain our dynamics from the second quantization of a one-body operator imposes a conservation law which rules out certain \textsc{ris} scenarios where nontrivial entropy production rates arise from interaction with a single reservoir; see e.g.\ the discussions surrounding Lemma~6.5 in~\cite{HJPR17} and Section~3.4 in~\cite{BB20}.

\section{Mixing}

We present several results on the large-time behaviour of the system. While explicit formulae using the canonical relations in Fock space have proved to be useful in~\cite{HJ17,Ra20}, we here focus on a scattering approach to the problem. We set
\begin{equation}
  \label{eq:def-Y0}
  Y_0 := C
\end{equation}
and
\begin{equation}
  \label{eq:def-Ym}
  Y_m := Z\BlSa M^{m-1} Z\SaBl
\end{equation}
for $m \geq 1$­. Heuristically, $Y_m$ encodes what happens to the wave function of a fermion from a reservoir which enters the sample, spends~$m-1$ more time steps there and then exits the sample.

\subsection{Scattering and the asymptotic state}
\label{sec:gen-i}

It is straightforward to check by induction that
\begin{multline}
  \label{eq:taking-powers}
  \fU^t
  - \sum_{n \neq 0,\dotsc,t-1} \delta_{n-t} \delta_{n}^* \otimes U^t \oplus 0
  \\  =
  \begin{pmatrix}
    \sum_{l = 0}^{t-1} \sum_{m=0}^{t-l-1} \delta_{l-t+m} \delta_l^* \otimes U^{t-l-m} Y_m U^l
    & \sum_{m=0}^{t-1} \delta_{-t+m} \otimes U^{t-m} Z\BlSa M^m
    \\
    \sum_{m=0}^{t-1} \delta_{t-m-1}^*\otimes  M^m Z\SaBl U^{t-m-1}
    & M^t
    \end{pmatrix}.
\end{multline}
for all~$t \geq 0$. As is customary, we investigate the behaviour of~$\fU^t$ for large~$t$ through M{\o}ller-like operators. Multiplying~\eqref{eq:taking-powers} by~$(S\otimes U \oplus \one)^{-t}$ on the right and performing a reindexation to eliminate explicit occurrences of~$t$ in the summand for the double sum, we find
\begin{multline}
\label{eq:pre-scat-t}
  \fU^t(S\otimes U \oplus \one)^{-t}
  - \sum_{n \neq -t,\dotsc,-1} \delta_{n} \delta_{n}^* \otimes \one \oplus 0
  \\
  =
  \begin{pmatrix}
    \sum_{m = 0}^{t-1}\sum_{l = 1}^{t-m} \delta_{-l} \delta_{-m-l}^* \otimes U^l Y_m U^{-m-l}
    & \sum_{m=0}^{t-1} \delta_{-t+m} \otimes U^{t-m} Z\BlSa M^m \\
    \sum_{m = 0}^{t-1}\delta^*_{-m-1} \otimes M^m Z\SaBl U^{-m-1}
    & M^t
  \end{pmatrix}
\end{multline}
for~$t \geq 0$. Multiplying the adjoint of~\eqref{eq:taking-powers} by~$(S\otimes U \oplus \one)^{t}$ on the right and performing a reindexation, we find a similar formula for $\fU^{-t}(S\otimes U \oplus \one)^{t}$ with $t \geq 0$.

Under Assumption~(Sp), it is thus easy to see from the matrix elements that the limits
\begin{equation}
\label{eq:scat-wlim}
  \Omega_U^\pm := \wlim\limits_{t \to \mp\infty} \fU^t(S\otimes U \oplus \one)^{-t}
\end{equation}
exist and are given by the explicit expressions
\begin{equation}\label{eq:scat}
\begin{split}
  \Omega_U^- &=
  \begin{pmatrix}
    \sum_{n \geq 0} \delta_n \delta_n^* \otimes \one
    & 0 \\
    0
    & 0
  \end{pmatrix}
  +
  \sum_{m \geq 0}
  \begin{pmatrix}
     \sum_{l \geq 1} \delta_{-l} \delta_{-m-l}^* \otimes U^l Y_m U^{-m-l}
      & 0 \\
   \delta^*_{-m-1} \otimes M^m Z\SaBl U^{-m-1}
      & 0
  \end{pmatrix}
\end{split}
\end{equation}
and
\begin{equation}\label{eq:scat-}
\begin{split}
  \Omega_U^+ &=
    \begin{pmatrix}
      \sum_{n' \leq -1} \delta_{n'} \delta_{n'}^* \otimes \one
      & 0 \\
      0
      & 0
    \end{pmatrix}
    +
    \sum_{m' \geq 0}
    \begin{pmatrix}
       \sum_{l' \geq m'} \delta_{l'-{m'}} \delta_{l'}^*  \otimes U^{m'-l'} Y_{m'}^* U^{l'}
        &  0 \\
        \delta_{m'}^* \otimes (M^*)^{m'} Z\BlSa^* U^{m'}
        & 0
    \end{pmatrix}.
\end{split}
\end{equation}
Note that we have not yet projected onto~$\ell^2(\zz) \otimes \Hb$, i.e.\ the subspace associated to the absolutely continuous spectrum of~$(S \otimes U \oplus \one)$, but have used the weak operator topology. As expected, strong convergence holds on the appropriate subspace; the proof of the following proposition concerning~$\Omega^-_U$ is postponed to Section~\ref{sec:scat-proofs}. While not needed in what follows, an analogue result holds for~$\Omega^+_U$.

\begin{proposition}\label{prop:scat}
  Suppose that Assumption~(Sp) holds. Then, both
  \begin{equation*}
  \label{eq:scat-slim}
    \slim\limits_{t \to \infty} \fU^t(S\otimes U \oplus \one)^{-t} (\one\otimes\one \oplus 0) = \Omega_U^- (\one\otimes\one \oplus 0)
  \end{equation*}
  and
  \begin{equation*}
        \slim\limits_{t \to \infty} (\fU^t(S\otimes U \oplus \one)^{-t})^* = (\Omega_U^-)^*.
  \end{equation*}
\end{proposition}

The scattering matrix
\[
  \fY_U := (\one\otimes\one \oplus 0)(\Omega_U^+)^* \Omega_U^-(\one\otimes\one \oplus 0)
\]
on~$\ell^2(\zz) \otimes \Hb$ will also frequently appear in the sequel. The following lemma makes its structure more explicit. A direct proof that~$\fY_U$ is unitary is given in the next section.

\begin{lemma}
  Under Assumption~(Sp),
  \begin{equation}
      \fY_U =\sum_{m \geq 0} \sum_{l\in\zz} \delta_l\delta_{l-m}^* \otimes U^{-l} Y_m U^{l-m}.
  \end{equation}
\end{lemma}

\begin{proof}
  We expand
  \begin{align*}
    &(\one\otimes\one \oplus 0)(\Omega_U^+)^* \Omega_U^-(\one\otimes\one \oplus 0)
    \\  & \qquad
    = \sum_{m\geq 0} \sum_{l \geq 1} \delta_{-l} \delta_{-m-l}^* \otimes U^l Y_m U^{-m-l} + \sum_{m'\geq 0} \sum_{l' \geq m'}  \delta_{l'} \delta_{l'-{m'}}^*  \otimes  U^{-l'} Y_{m'} U^{-m'+l'}
      \\ & \qquad \qquad \qquad
      + \sum_{m' \geq 0} \sum_{m \geq 0} \delta_{m'}\delta_{-m-1}^* \otimes U^{-m'} Z\BlSa M^{m'} M^m Z\SaBl U^{-m-1}.
  \end{align*}
  Rewriting the double sum on the last line in terms of~$Y_{m''}$ with~$m'' = m + m'$ yields the desired formula.
\end{proof}

  We use a subscript~$U$ on some of the objects introduced in this section because it is at times convenient to factor out the contribution from the unitary~$U$ and then consider the special case~$U = \one$. For example,
  \[
    \Omega_{U}^{-} = \bigg(\sum_{m \in \zz} P_m \otimes U^m \oplus \one \bigg)^* \Omega_{\one}^- \bigg(\sum_{n \in \zz} P_n \otimes U^n \oplus \one \bigg)
  \]
  and
  \begin{equation}
  \label{eq:fY-wo-U}
    \fY_{U}=
      \Big(\sum_{m \in \zz} P_m \otimes U^m \Big)^* \fY_{\one} \Big(\sum_{n \in \zz} P_n \otimes U^n \Big).
  \end{equation}
In view of this factorization, we introduce a modification of~$\T$ which absorbs part of the free dynamics in the environment:
\begin{equation}
\label{eq:def-Xi}
  \Xi := \bigg( \sum_{n \in \zz} P_n \otimes U^n \bigg) \T \bigg( \sum_{m \in \zz} P_m \otimes U^m \bigg)^*,
\end{equation}
so that
\[
  \Xi = \sum_{n,m \in \zz} \delta_n \delta_m^* \otimes \Xi_{m-n},
\]
where
\begin{equation*}
  \Xi_n := \T_{0,n} U^{-n}.
\end{equation*}

Note that~$\Xi$ is selfadjoint and commutes with~$S \otimes \one$ and $\one \otimes \Pi_k$, $k=1,\dots, n_B$.

\begin{proposition}
  Under Assumptions~(IC) and (Sp), the limit
  \begin{equation}
    \rho(\infty)[A] := \lim_{t \to \infty} \rho(0)[\Gamma(\fU)^{-t} A \Gamma(\fU)^{t}]
  \end{equation}
  exists for all~$A \in \CAR(\H\tot)$ and defines a quasifree state with symbol
  \begin{equation}
    \T\tot^\infty := \Omega_U^-  (\T \oplus 0) (\Omega_U^-)^*.
  \end{equation}
\end{proposition}

\begin{proof}
  To prove the proposition it suffices to show that
  \[
    \lim_{t\to\infty}\rho(0)\left[\Gamma(\fU^*)^t \Big(\prod_{h=1}^N a(V_h)\Big)^* \Big(\prod_{h'=1}^{N'} a(V'_{h'})\Big) \Gamma(\fU)^t \right] = \delta_{N,N'} \det [\braket{V'_{h'}, \T\tot^\infty V_{h}}]_{h,h'=1}^N
  \]
  for an arbitrary choice of~$N,N' \geq 0$ and~$V_1, \dotsc, V_N, V'_1, \dotsc, V'_{N'} \in \H\tot$. Because~$\T$ commutes with~$S \otimes U$, we have
  \[
    \rho(0)[A] =   \rho(0)\big[\Gamma(S \otimes U \oplus \one)^t A \Gamma(S^* \otimes U^* \oplus \one)^t\big]
  \]
  for all~$A \in \CAR(\H\tot)$ and the Bogolyubov relation gives that the identity to be shown is equivalent to
  \begin{equation}
  \label{eq:int-st-Bog}
    \lim_{t\to\infty}\rho(0)\left[ \Big(\prod_{h=1}^N a((\Omega_U^{(t)})^*V_h)\Big)^* \prod_{h'=1}^{N'} a((\Omega_U^{(t)})^*V'_{h'}) \right] = \delta_{N,N'} \det [\braket{V'_{h'}, \T\tot^\infty V_{h}}]_{h,h'=1}^N,
  \end{equation}
	where
	\[
		\Omega_U^{(t)} := \fU^t(S\otimes U \oplus \one)^{-t}.
	\]
  First note that
  \[
    \lim_{t \to \infty} \| (\Omega_U^{(t)})^*V_h - (\one \otimes \one \oplus 0)(\Omega_U^{(t)})^*V_h \| = 0
  \]
  for each $h = 1,\dotsc, N$ by Proposition~\ref{prop:scat}, and similarly with primes. Hence, by continuity of the fermionic creation and annihilation operators as functions from~$(\H\tot,\|\,\cdot\,\|)$ to~$(\mathcal{B}(\Ga(\H\tot)),\|\,\cdot\,\|)$, the limit in~\eqref{eq:int-st-Bog} will exist if and only if the limit
  \begin{align*}
    \lim_{t\to\infty}\omega_{\T}\left[ \Big(\prod_{h=1}^N a((\one \otimes \one \oplus 0)(\Omega_U^{(t)})^*V_h)\Big)^* \prod_{h'=1}^{N'} a((\one \otimes \one \oplus 0)(\Omega_U^{(t)})^*V'_{h'}) \right]
  \end{align*}
  exists, in which case they will coincide. In particular, we may as well assume that the initial state~$\rho\Sa(0)$ is quasifree with vanishing symbol.

  Under this extra assumption, the state~$\rho(t)$ is quasifree for all~$t \in \nn$ and has symbol $\T\tot(t)$:
  \begin{equation*}
    \rho(0)\left[ \Big(\prod_{h=1}^N a((\Omega_U^{(t)})^*V_h)\Big)^* \prod_{h'=1}^{N'} a((\Omega_U^{(t)})^*V'_{h'}) \right] = \delta_{N,N'} \det [\braket{V'_{h'}, \T\tot(t) V_{h}}]_{h,h'=1}^N,
  \end{equation*}
  where
  \begin{align*}
    \T\tot(t) &= \Omega_U^{(t)}(\T \oplus 0)(\Omega_U^{(t)})^*.
  \end{align*}
  Therefore, we will be done if we can show that~$\T\tot(t)$ converges weakly to the proposed limit~$\T\tot^\infty$. But this is easily deduced from Proposition~\ref{prop:scat}.
\end{proof}

We are now in a position to get the symbol of the restriction of the state to the sample, i.e.\
\begin{equation}
\label{eq:def-D-infty}
  \Delta^\infty := (0 \oplus \one)\T\tot^\infty(0 \oplus \one).
\end{equation}

\begin{proposition}\label{prop:Delta}
  Suppose that Assumptions~\textnormal{(Sp)} {and \textnormal{(IC)}} hold and let
  \[
    \Psi(X) := \sum_{k = 0}^\infty M^k X (M^*)^k
  \]
  for~$X: \Hs \to \Hs$. Then,
  \begin{align*}
    \Delta^\infty = \Psi(G+G^*),
  \end{align*}
  where
  \[
    G := \tfrac 12 Z\SaBl \Xi_0 Z\SaBl^* + \sum_{l = 1}^\infty M^l Z\SaBl \Xi_l Z\SaBl^*.
  \]
\end{proposition}

\begin{proof}
  Since,
  \[
    (0 \oplus \one) \Omega_{U}^- (\one \oplus 0) = \sum_{m=0}^\infty \delta_{-m-1}^* { \otimes} M^m Z\SaBl U^{-m-1}
  \]
  by Proposition~\ref{prop:scat}, \eqref{eq:def-D-infty} gives
  \begin{align*}
    \Delta^\infty &= \sum_{m,n=0}^\infty  M^m Z\SaBl U^{-m-1} \T_{-m-1,-n-1} U^{n+1} Z\SaBl^* (M^*)^n \\
      &= \sum_{m,n=0}^\infty M^m Z\SaBl \Xi_{m-n} Z\SaBl^*(M^*)^n
  \end{align*}
  using $\T_{-m-1,-n-1} = U^{m+1}\Xi_{m-n}U^{-n-1}$. Splitting the contributions with $m-n > 0$, $m-n = 0$ and~$m-n < 0$ and reindexing with $l=|m-n|$ gives the proposed formula.
\end{proof}

 \begin{remark}\label{rem33bis}
   If Assumption~\textnormal{(IC+)} holds, the symbol of the restriction to the sample at time $t$ reads
   \[
    \Delta^t = M^t \Delta (M^*)^t + \sum_{m=0}^{t-1}\sum_{n=0}^{t-1}  M^m Z \SaBl \T_{0,m-n} U^{n-m} Z\SaBl^* (M^*)^n.
   \]
 \end{remark}

We now turn our attention to the block
\[
  \T\En^\infty := (\one\otimes\one \oplus 0) \T\tot^\infty (\one\otimes\one \oplus 0)
\]
of~$\T\tot$ corresponding to the environment. As a direct consequence of Proposition~\ref{prop:scat}, we have the following corollary.
\begin{corollary}
  Suppose that Assumption~\textnormal{(Sp)} holds and let $\rho(0)$ be an initial state on~$\Ga(\Htot)$ as in Assumption~\textnormal{(IC)}. Then,
  \begin{equation}
  \label{eq:blocks-T-infty}
    \delta^*_n \T^\infty\En \delta_m
    =
    \begin{cases}
      U^{-n} \left( \sum_{l,l' \geq 0} Y_l \Xi_{l-l'+m-n} Y_{l'}^* \right) U^m
      & n < 0, m < 0,  \\
      U^{-n} \left( \sum_{l \geq 0} Y_l \Xi_{l+m-n} \right) U^m
      & n < 0, m \geq 0, \\
      U^{-n} \Xi_{m-n} U^m
      & n \geq 0, m \geq 0. \\
    \end{cases}
  \end{equation}
  In particular, $\delta^*_n \T^\infty\En \delta_m = \delta^*_n \T \delta_m$ for~$n,m \geq 0$.
\end{corollary}

Note that the asymptotic symbol $\T\En^\infty$ need not commute with~$S \otimes U$; blocks corresponding to positions having already interacted (negative indices) are given a different expression than those corresponding to position which have not yet interacted. This is inherent to our choice of dynamics in the environment, which prevents the effects of the interaction taking place at the site zero to affect the state at locations that have not yet been in contact with the sample. We will come back to this point in the next subsection.

\subsection{Fourier representation}
\label{sec:Fourier}

Many of the expressions call for a representation in Fourier space that we will take advantage of in what follows. We introduce the unitary map
$
  \mathcal{F} : \ell^2(\zz) \otimes \Hb \to L^2([0,2\pi];\Hb)
$
as follows: for $\psi = \sum_{l \in \zz} \delta_l \otimes \psi_l$ {with $\sum_{l \in \zz} \|\psi_l\|^2 < \infty$} and $\theta \in [0,2\pi]$, we set
\[
  (\mathcal{F}\psi)(\theta) := \sum_{l \in \zz} \Exp{{ - \ii l \theta}} \psi_l.
\]
In practice, we will more often use the notation
\[
  \hat\psi := \mathcal{F}\psi.
\]
Let $R : \ell^2(\zz) \otimes \Hb \to \ell^2(\zz) \otimes \Hb$ have the form
\begin{equation}
\label{eq:form-tb-Fourier}
  R = \sum_{n,m \in \zz} \delta_n \delta_m^* \otimes R_{m-n}
\end{equation}
for some norm-summable sequence~$(R_l)_{l \in \zz}$ of operators on~$\Hb$\,---\,hereafter referred to as Fourier coefficients\,---, so that~$\|R\| \leq \sum_{l \in \zz} \|R_l\|$. Then,
\[
  (\mathcal{F}R\psi)(\theta) = \big((\mathcal{F}R\mathcal{F}^{-1})(\mathcal{F}\psi)\big)(\theta)
  = \hat{R}(\theta) \hat\psi(\theta),
\]
where $\hat{R} : L^2([0,2\pi];\Hb) \to L^2([0,2\pi];\Hb)$ is the multiplication operator by
\[
  \hat{R}(\theta) := \sum_{l \in \zz} \Exp{\ii l \theta} R_l.
\]
Also note that~$R$ is selfadjoint if and only if~$R_{-l} = R_l^*$  for each~$l \in \zz$, in which case~$\hat{R}(\theta)$ is { selfadjoint} for all~$\theta\in[0,2\pi]$.

We will make use of this representation for~$\Xi$:
\begin{align*}
  \hat{\Xi}(\theta) &= \sum_{l \in \zz} \Exp{\ii l \theta} \Xi_l.
\end{align*}
Recall That~$\Xi$ is of the form~\eqref{eq:form-tb-Fourier} by construction (under Assumption~(IC)), with blocks
\begin{equation}
\label{eq:blocks-Xi}
	\Xi_{m-n} = U^n T_{n,m}U^{-m}.
\end{equation}
Then, with
\[
	\hat{\fY}(\theta) := \sum_{l \geq 0} \Exp{-\ii l \theta} Y_l
\]
(note the sign of~$\ii l \theta$), we set
\[
	\hat{\Xi}^\infty(\theta) := \hat{\fY}(\theta) \hat{\Xi}(\theta) \hat{\fY}(\theta)^*.
\]
Equivalently, $\hat{\Xi}^\infty(\theta)$ is the Fourier representation of an operator~$\Xi^\infty$ of the form~\eqref{eq:form-tb-Fourier} with blocks
\begin{equation}
\label{eq:blocks-Xi-infty}
	\Xi^\infty_{m} = \sum_{l,l'\geq 0} Y_l \Xi_{l-l'+m} Y_{l'}^*
\end{equation}
for all $m \in \zz$. To see this, integrate $\hat{\fY}(\theta) \hat{\Xi}(\theta) \hat{\fY}(\theta)^*$ against~$\tfrac{1}{2\pi}\Exp{-\ii m \theta}$ to find the~$m$-th block.

Note that combining~\eqref{eq:blocks-T-infty} and~\eqref{eq:blocks-Xi-infty} gives
\begin{equation}
\label{eq:blocks-Xi-infty-diff}
	\Xi^\infty_{m-n} = U^n \delta_n^* \T\En^\infty \delta_m U^{-m}
\end{equation}
if~$n < 0$ and~$m < 0$. In other words, $\Xi^\infty$ is translation invariant, but as far as blocks that have been affected by the interaction with the sample, $\Xi^\infty$ is to~$\T\En^\infty$
as~$\Xi$ is to~$\T$; compare~\eqref{eq:blocks-Xi-infty-diff} to~\eqref{eq:blocks-Xi}. Note that $\T\En^\infty =  \T$ implies~$\Xi^\infty = \Xi$; the converse implication fails.

\begin{lemma}
  The operator~$\fY_{U}$ is unitary.
\end{lemma}

\begin{proof}
  In view of~\eqref{eq:fY-wo-U} and \eqref{eq:def-Xi}, it suffices to prove the lemma with~$U = \one$. Let
  \[
    \hat{\fY}(\theta) := \sum_{l \geq 0} \Exp{-\ii l \theta} Y_l
  \]
  be as in the previous { discussion}; it is clear that it suffices to show that~$\hat{\fY}(\theta)$  { is unitary} for all~$\theta \in \rr$.
  Given the definitions~$Y_0 := C$ and~$Y_l := Z\BlSa M^{l-1} Z\SaBl$ for~$l \geq 1$, the operator~$\hat{\fY}(\theta)$ can be expressed in terms of resolvents of~$M$:
  \begin{equation}
    \hat{\fY}(\theta) = C + \sum_{l \geq 0} \Exp{-\ii \theta} \Exp{-\ii l \theta} Z\BlSa M^l Z\SaBl = C - Z\BlSa (M - \Exp{\ii \theta})^{-1} Z\SaBl
  \end{equation}
  an expression which is well defined for all~$\theta \in \rr$ under Assumption~(Sp). The operators involved correspond to the block representation~\eqref{eq:def-Z} of the unitary operator~$Z$. Unitarity of~$\hat{\fY}$ is given by the next lemma and the present lemma follows.
\end{proof}

\begin{lemma}
  Let~$Z$ be a unitary operator with block decomposition~$Z = (\begin{smallmatrix} a & b \\ c & d \end{smallmatrix})$ with respect to an orthogonal direct sum decomposition of a finite-dimensional Hilbert space. Then, for all~$\eta \in \rr$, the bounded operator $s(\eta) := a - b(d-\Exp{\ii \eta})^{-1}c$ is a unitary operator on the first subspace in the decomposition.
\end{lemma}

\begin{proof}
  Simply expand the expression~$s(\eta)s(\eta)^*$ and make use of the relation satisfied by~$a,b,c$ and~$d$ as a consequence of unitarity of~$Z$ as well as of the identity $d(d-\Exp{\ii \eta})^{-1} = \one + \Exp{\ii \eta} (d-{\Exp{\ii \eta}})^{-1}$.
\end{proof}

\section{Fluxes of particles}

We associate to a bounded selfadjoint operator~$X : \Hb \to \Hb$ the flux
\[
  \Phi_X = \dG(\fU^* (\one\otimes X \oplus 0)\fU - \one \otimes X \oplus 0).
\]
Using the block form of~$\fU$, {and assuming $[X,U]=0$}, one can check that
\[
  \fU^* (\one\otimes X \oplus 0)\fU - \one \otimes X \oplus 0 =
  \begin{pmatrix}
    { -P_0} \otimes X + P_0 \otimes C^* X C           & \delta_0 \otimes C^* X Z\BlSa \\
    \delta_0^* \otimes Z\BlSa^* X C  &  Z\BlSa^* X Z\BlSa
  \end{pmatrix}
\]
is trace class. The interest of such quantities is best seen through the case of particle fluxes between the different parts of the environment, hereafter referred to as \emph{reservoirs}, whose definition requires the structure in Assumption~(Bl). Such a structure is evidently present in the special case discussed in the introduction. Formally, the (infinite) number of fermions in the reservoir~$\ell^2(\zz) \otimes \Pi_k\Hb$ is given by the observable~$\dG(\one\otimes \Pi_k \oplus 0)$, { where $[\Pi_k,U]=0$,} and the number of fermions that enter this reservoir in one time step is given by the observable
\[
  \Phi_k \equiv \Phi_{\Pi_k} = \Gamma(\fU^*) \dG(\one\otimes \Pi_k \oplus 0) \Gamma(\fU) - \dG(\one\otimes \Pi_k \oplus 0)
\]
on~$\Ga(\Htot)$.

Back to the general observable~$X$ {such that $[X,U]=0$}, we know from Section~\ref{sec:gen-i} that the asymptotic state of the full system, denoted~$\rho(\infty)$, is quasifree with symbol~$\T\tot^\infty = \Omega_{U}^-(\T\oplus 0) (\Omega_{U}^-)^*$ if~$\rho(0)$ satisfies Assumption~(IC). Hence, the steady-state expectation value of the flux~$\Phi_X$, or current, is given by
\begin{equation}
\begin{split}
  J_X &:= \rho(\infty)[\Phi_X] \\
    &\phantom{:}= \tr_{\Htot} [\T\tot^\infty \{\fU^* (\one \otimes X \oplus 0) \fU - \one \otimes X \oplus 0\}].
\end{split}
\end{equation}
Using the decomposition
\begin{equation}
  \T\tot^\infty
  =
  \begin{pmatrix}
    \T^\infty\En & \T^\infty\EnSa \\
    \T^\infty\SaEn & \Delta^\infty
  \end{pmatrix},
\end{equation}
we get
\begin{equation}
\begin{split}
  J_X &= \tr_{\ell^2(\zz)\otimes\Hb}( \T^\infty\En(P_0 \otimes (CXC^* - X))+\T^\infty\EnSa (\delta_0^* \otimes Z^*\BlSa XC) )\\
    &\qquad + \tr_{\Hs}(\T^\infty\SaEn (\delta_0 \otimes C^*XZ\BlSa + \Delta^\infty Z^*\BlSa X Z\BlSa).
\end{split}
\end{equation}
This expression serves as a basis for obtaining more transparent expressions.

\begin{proposition}\label{prop:flux_X}
  Under Assumptions~\textnormal{(IC)} and~\textnormal{(Sp)}, if~$X: \H\Bl \to \H\Bl$ is a bounded observable  {such that $[X,U]=0$}, then
  \[
    J_X = \tr \bigg[  X \int_0^{2\pi} (\hat{\fY}(\theta) \hat{\Xi}(\theta) \hat{\fY}^*(\theta) - \hat{\Xi}(\theta)) \frac{\d\theta}{2\pi} \bigg].
  \]
\end{proposition}

\begin{proof}[Proof sketch]
  We consider the case~$U = \one$ to lighten the notation. Use cyclicity of the trace to rewrite the trace over~$\Hs$ as a trace over~$\ell^2(\zz) \otimes \Hb$. Then, expand the formulae for~$\T^\infty\En$, $\T^\infty\EnSa$, $\T^\infty\SaEn$ and~$\Delta^\infty$. The part on~$\ell^2(\zz)$ is restricted to the span of~$\delta_0$ and we are left with a trace on~$\Hb$. Rewrite this trace gathering all occurrences of~$Y_m$ defined by~\eqref{eq:def-Y0}--\eqref{eq:def-Ym}:
  \begin{equation}\label{eq:JofX}
    J_X = \tr_{\Hb} \bigg[  X\bigg(\sum_{n,m\geq 0} Y_n \Xi_{n-m} Y_m^*  - \Xi_0\bigg) \bigg].
  \end{equation}
  Conclude using the identity~\eqref{eq:blocks-Xi-infty}.
\end{proof}

For the currents~$J_k \equiv J_{\Pi_k}$ associated to the projectors~$\Pi_k$, $k = 1, \dotsc, n\Bl$, we immediately get the two following consequences.

\begin{corollary}\label{cor:avg-curr}
  Under Assumptions~\textnormal{(IC)},~\textnormal{(Sp)} and~\textnormal{(Bl)}, we have
  \[
    \sum_{k=1}^{n\Bl} J_k = 0.
  \]
  More precisely, for each~$k=1,\dotsc,n\Bl$,
  \[
    J_k = \sum_{k' \neq k} \int \tr[\hat{\fY}^*(\theta) \Pi_k \hat{\fY}(\theta) \Pi_{k'} \hat{\Xi}(\theta)  ] - \tr[\hat{\fY}^*(\theta) \Pi_{k'} \hat{\fY}(\theta) \Pi_k\hat{\Xi}(\theta) ] \, \frac{\d\theta}{2\pi}
  \]
  and, with the additional assumption that each~$\{\Pi_k\}_{k=1}^{n\Bl}$ has rank one,
  \begin{equation}
    \label{eq:current-as-int}
    J_k = \int \sum_{k' \neq k} C_{k,k'}(\theta) (f_{k'}(\theta) - f_k(\theta)) \, \frac{\d\theta}{2\pi},
  \end{equation}
  where~$f_k(\theta) := \tr[\Pi_{k} \hat{\Xi}(\theta)]$ and $C_{k,k'}(\theta) := \tr[\hat{\fY}^*(\theta) \Pi_k \hat{\fY}(\theta) \Pi_{k'}]$ are nonnegative,  and satisfy $$\sum_{k'=1}^{n\Bl} C_{k,k'}(\theta) =\sum_{k=1}^{n\Bl} C_{k,k'}(\theta) = 1.$$
\end{corollary}

\begin{remark}\label{rem:fluxes_ajpp}
  Formula~\eqref{eq:current-as-int} in the case where each~$\Pi_k$ has rank one implies in particular that if one of the functions $f_k : \theta\mapsto \tr[\Pi_{k}\hat{\Xi}(\theta)]$ satisfies $f_k(\theta)\geq f_{k'}(\theta)$ for all $k'\neq k$, then the flux of particles is necessarily going out of the $k$-th reservoir (i.e. $J_k\leq 0$).
\end{remark}

\begin{remark}
  We may think of the~$C_{k', k}(\theta)$ as some effective conductance at frequency~$\theta$. This is similar to the Landauer--B{\"u}ttiker formula presented in \cite{AJPP07} (Corollary~4.2), with the following differences: the context in \cite{AJPP07} is in continuous time and not in discrete time, and the free dynamics on the reservoir number~$k$ is generated by some Hamiltonian~$h_k$ instead of the shift~$S$.
  The flux of some observable $q$ is then expressed as a sum of integrals over $\operatorname{sp}_{\textnormal{ac}}(h_k) \cap \operatorname{sp}_{\textnormal{ac}}(h_{k'})$, where $\operatorname{sp}_{\textnormal{ac}}(h_{k'})$ is the absolutely continuous spectrum of the Hamiltonian~$h_{k'}$ of another reservoir, while in our expression we integrate over the spectrum of~$S$, i.e.\ the unit circle.
\end{remark}

\begin{proof}[Proof of Corollary~\ref{cor:avg-curr}]
  We have $J_k=J_{\Pi_k}$, which by Proposition~\ref{prop:flux_X} gives
  \begin{align*}
    J_k
      &= \int_0^{2\pi} \tr \big[  \Pi_k \hat{\fY}(\theta) \hat{\Xi}(\theta) \hat{\fY}^*(\theta) -\Pi_k \hat{\Xi}(\theta)\big] \frac{\d\theta}{2\pi}.
  \end{align*}
  Since $\sum_{k'=1}^{n_B} \Pi_{k'}=\one$  and  $\hat{\fY}(\theta)$ is unitary for all $\theta$, we have $\sum_{k=1}^{n_B} J_k=0$.
  Now by assumption \textnormal{(Bl)} we have $\Xi=\sum_{k'=1}^{n_B} \Pi_{k'} \Xi \Pi_{k'}=\sum_{k'=1}^{n_B} \Pi_{k'} \Xi$ hence
  \begin{align*}
    J_k
      &=\int_0^{2\pi} \tr \Big[   \sum_{k'=1}^{n_B} \Pi_k \hat{\fY}(\theta) \Pi_{k'}\hat{\Xi}(\theta) \hat{\fY}^*(\theta) -\Pi_k \hat{\Xi}(\theta)\Big] \frac{\d\theta}{2\pi}
  \end{align*}
  and by the properties of $\Pi_k$ and $\hat{\fY}(\theta)$ we have
  \[
    \tr\big[ \Pi_k \hat{\Xi}(\theta) \big]
      =\sum_{k'=1}^{n_B} \tr\big[\hat{\fY}^*(\theta) \Pi_{k'} \hat{\fY}(\theta) \Pi_k\hat{\Xi}(\theta)\big].
  \]
  This proves that $J_k=\sum_{k'\neq k} A_{k, k'}-A_{k', k}$ for $A_{k', k}=\tr[\hat{\fY}^*(\theta) \Pi_k \hat{\fY}(\theta) \Pi_{k'} \hat{\Xi}(\theta)  ]$.
  Moreover, in the case where the $\Pi_k$ are of rank one, we have
  \[
    \Pi_k \hat{\Xi}(\theta)=\Pi_k \hat{\Xi}(\theta) \Pi_k=\tr[\Pi_k \hat{\Xi}(\theta)]\,\Pi_k
  \]
  and, restoring the summation to all indices,
  \[
    \sum_{k'} \tr \big[   \hat{\fY}^*(\theta)  \Pi_{k'} \hat{\fY}(\theta)  \Pi_{k} \big]=\tr[  \Pi_{k}]=\sum_{k'} \tr \big[   \hat{\fY}(\theta)  \Pi_{k'} \hat{\fY}^*(\theta)  \Pi_{k} \big],
  \]
  which gives the second formula for $J_k$ and the summation property of $C_{k,k'}(\theta)$.
\end{proof}

\section{Entropy production}

Since nontrivial asymptotic currents can develop between the reservoirs of the system at hand, we expect that the total system genuinely settles into a nonequilibrium steady state. Another key signature of such states is the nontrivial entropy production rate they give rise to.  We prove here the existence and strict positivity of the asymptotic entropy production rate related to the convergence towards the nonequilibrium steady state.
More precisely, we work under Assumption~(IC+) and provide a convergence result for the quantity
\begin{equation}\label{eq:rel-ent-qf}
  \sigma(t) := t^{-1}\big(S[\T\tot(t)|\T\tot(0)] + S[\one- \T\tot(t)|\one - \T\tot(0)]\big),
\end{equation}
where $\T\tot(t) := \Omega_U^{(t)}(\T \oplus \Delta) (\Omega_U^{(t)})^*$ and
\[
	S[X|Y] := \tr[X (\log X - \log Y)]
\]
for any trace-class operators~$X$ and~$Y$ with~$\epsilon \leq X,Y \leq \one - \epsilon$ on some common Hilbert space. This definition is motivated by a formula for the relative entropy between quasifree states which is well established for finite-dimensional systems~\cite[\S{IV.B}]{DFP08} and the observation that~$\Omega_U^{(t)}$ is a finite-rank perturbation of the identity. It will also be {a posteriori} justified by the relation to fluxes established in Corollary~\ref{sigmuj}.

The following theorem states that the entropy production rate converges to the integral of the relative entropies of matrices related to the initial and asymptotic states of the environment introduced in Section~\ref{sec:Fourier}. Its proof is postponed to Section~\ref{sec:ep-proof}.

\begin{theorem}\label{thm:lim-EP}
  Under Assumption~\textnormal{(IC+)}, $\T\tot(t) - \T\tot(0)$ has finite rank and~$\sigma(t)$ in~\eqref{eq:rel-ent-qf} is well defined for all~$t \in \nn$. If, in addition, Assumption~\textnormal{(Sp)}
	holds, then the limit
  \[
    \sigma^+ := \lim_{t \to \infty} \sigma(t)
  \]
  exists and is given by
  \begin{equation}\label{eq:ent_prod_integral_general}
    \sigma^+ = \int_0^{2\pi} S\big[\hat{\fY}(\theta)\hat{\Xi}(\theta)\hat{\fY}^*(\theta)\big| \hat{\Xi}(\theta)\big]  \, \frac{\d\theta}{2\pi} + \int_0^{2\pi} S\big[\one - \hat{\fY}(\theta)\hat{\Xi}(\theta)\hat{\fY}^*(\theta)\big|\one -  \hat{\Xi}(\theta)\big]  \, \frac{\d\theta}{2\pi}.
  \end{equation}
  Moreover, $\sigma^+ \geq 0$ with equality if and only if~$\hat{\Xi}(\theta) = \hat{\fY}(\theta)\hat{\Xi}(\theta)\hat{\fY}^*(\theta)$ for Lebesgue-almost all~$\theta \in [0,2\pi]$.
\end{theorem}

\begin{remark}
  Recall that $\theta \mapsto \hat{\fY}(\theta)\hat{\Xi}(\theta)\hat{\fY}^*(\theta)$ is the Fourier transform of a  translation-invariant operator~${\Xi}^\infty$ which, up to the transformation which relates~$\Xi$ to~$\T$, shares its blocks with~$\T\En^\infty$.
\end{remark}

The following reformulation of the result is closer to typical formulations in terms of currents and thermodynamic potentials (see for example Equation (17) in \cite{JPW14}), albeit frequency-wise.  It can be compared to Corollary 4.3 of \cite{AJPP07}; see also Remark~\ref{rem:fluxes_ajpp}.

\begin{corollary}\label{sigmuj}
  Suppose that Assumptions~\textnormal{(IC+)},~\textnormal{(Sp)} and~\textnormal{(Bl)} hold with the projectors~$\Pi_1,\dotsc,\Pi_{n\Bl}$ having rank one.
  Then we have the identity
  \begin{equation}
    \sigma^+ = \sum_{k=1}^{n\Bl} \int_0^{2\pi} \mu_k(\theta) \hat\jmath_k(\theta)\,\frac{\d\theta}{2\pi},
  \end{equation}
  where
  \[
    \mu_k(\theta) := \log \frac{1-f_k(\theta)}{f_k(\theta)}
  \]
  and~$\hat\jmath_k(\theta)$ denotes the integrand of the expression~\eqref{eq:current-as-int} for the~$k$-th flux of particles.
\end{corollary}

\begin{remark}
  In the case where each~$f_k$ is constant in~$\theta$, the formula simplifies to
  \begin{align*}
    \sigma_+
      &= \sum_{k=1}^{n\Bl}\sum_{k'=1}^{n\Bl} \mu_k  (f_{k'} - f_k) \int C_{k,k'}(\theta) \, \frac{\d\theta}{2\pi} \\
      &= \sum_{k=1}^{n\Bl}\sum_{k'=1}^{n\Bl} \mu_k  (f_{k'} - f_k) \sum_{l \geq 0} \tr[Y_l^* \Pi_k Y_l \Pi_{k'}].
  \end{align*}
 \end{remark}

At this stage, the picture of entropy production is still short of a study of the statistical fluctuations in  measurement processes of physical observable properly related to the ``information-theoretical'' notion of entropy production; see e.g.~\cite[\S{4.4.5}]{JOPP11}.

\section{Discussion for small coupling strength}
\label{sec:small}

In order to investigate the regime where the interaction between the sample and its environment is weak, we will consider a special case where the unitary operator~$Z$ on $\Hb\oplus \Hs$ is of the form
\[
  Z
    =
    \begin{pmatrix}
        \one & 0 \\ 0 & W
    \end{pmatrix}
    \exp\left[-\ii \alpha \begin{pmatrix} 0 & A^* \\ A & 0 \end{pmatrix} \right]
\]
for some unitary operator~$W : \Hs\to \Hs$ which represents the free evolution on the sample, some bounded operator~$A : \Hb \to \Hs$ which couples sites of the sample and sites of the environment, and some coupling strength
$\alpha\in \rr$. Computing the exponential, we obtain
\begin{align*}
  C &= \cos (\alpha \sqrt{A^*A})
    & Z\BlSa &= -\ii A^* \frac{\sin (\alpha \sqrt{AA^*})}{\sqrt{AA^*}}  \\
  Z\SaBl &= - \ii  W \frac{\sin(\alpha \sqrt{AA^*})}{\sqrt{AA^*}} A
    & M &= W \cos(\alpha \sqrt{AA^*}).
\end{align*}
In this particular setup, we can give a more tractable condition for the Assumption~(Sp) to hold true as well as more explicit formulas as the coupling strength~$\alpha$ tends to~$0$.

\begin{proposition}\label{prop:Kal}
  Let us consider $M(\alpha) := W\cos(\alpha\sqrt{AA^*})$ for $\alpha\in \rr$, and write $\mathcal{V}\subseteq \Hs $ the range of $A$. Then, there exists~$\alpha_A > 0$ depending on~$A$ only such that the following properties are equivalent:
  \begin{enumerate}
    \item The spectrum of $M(\alpha)$ is contained in the interior of the unit disc for all~$\alpha \in (-\alpha_A,\alpha_A)$,
    \item The subspace $\mathcal{V}$ is contained in no strict {subspace of~$\Hs$} which is stable by~$W$,
    \item We have
    \[
      \operatorname*{span}_{i=0, \dotsc, \dim \Hs} W^i \mathcal{V} = \Hs,
    \]
  \end{enumerate}
\end{proposition}

The equivalence between the second and third property is well known and only included because of its relation to linear control theory, where it is called the Kalman condition.

\begin{proof}
  Let $\{\mu_i\}_{i\geq 0}$ be the (nonnegative) eigenvalues of~$\sqrt{AA^*}$ and let~$\{p_i\}_{i \geq 0}$ be the corresponding spectral projectors. We include~$0$ as~$\mu_0$, possibly at the cost of having~$p_0 = 0$. Choose~$\alpha_A > 0$ small enough that $|\alpha \mu_i| < \pi$ whenever~$|\alpha| < \alpha_A $. Then, with~$\nu_i := \cos(\alpha \mu_i)$, we have
  \[
  \cos(\alpha \sqrt{AA^*})=p_0+\sum_{i=1}^l \nu_i p_i.
  \]
  Note that~$p_0$ is the orthogonal projection onto the kernel of~$\sqrt{AA^*}$, which coincides with the orthogonal complement of~$\mathcal{V}$.

  If the first property is not satisfied, then there exists a normalized eigenvector~$\phi$ of~$M(\alpha)$ with eigenvalue~$\lambda$ with $|\lambda| \geq 1$ for some~$\alpha \in (-\alpha_A,\alpha_A)$. Then,
  \[
    |\lambda|^2
      =\braket{M(\alpha)\phi, M(\alpha)\phi}
      =\braket{\phi, p_0\phi}+\sum_{i\geq 1} \nu_i^2 \braket{\phi, p_i\phi}
  \]
  and since $\sum_{i\geq 0} \braket{\phi, p_i\phi}=1$ this implies that $|\lambda|^2=1$, $p_0 \phi = \phi$ and $\sum_{i \geq 1} p_i \phi = 0$. Then,~$\phi$ is in the orthogonal complement of~$\mathcal{V}$ and is also an eigenvector of~$W$ since
  \[
    \lambda \phi = M(\alpha)\phi = W \bigg(p_0+\sum_{i\geq 1} \nu_i p_i\bigg)\phi = W\phi.
  \]
  We conclude that~$\mathcal{V}$ is contained in the orthogonal complement of the span of~$\phi$, which is stable by~$W$ since~$\phi$ is an eigenvector of~$W$. Thus the second property is not satisfied.

  Conversely, if the second property is not satisfied, then there exists an eigenvector~$\phi$ of~$W$ in the orthogonal complement of~$\mathcal{V}$. Then,~$\phi$ is clearly an eigenvector of~$M(\alpha)$ with eigenvalue on the unit circle for all~$\alpha$, which implies in particular that the first property is not satisfied.
\end{proof}

In order to carry some usual procedures from perturbation theory, we will need a semisimplicity and regularity assumption on the spectral decompostion of the family  of operators~$M(\alpha)$  analytic in the coupling strength~$\alpha$.

\begin{description}
  \item[Assumption ($\tfrac 12$Sim)] There exists a punctured neighbourhood~$\Omega$ of~$0$ in~$\cc$ such that the eigenvalues of~$M(\alpha)$ are semisimple for all~$\alpha \in \Omega$ and there is a decomposition
  \begin{equation}
  \label{eq:spec-M}
    M(\alpha) = \sum_{j \in I} \lambda_j(\alpha) Q_j(\alpha)
  \end{equation}
  with scalar functions $\lambda_j : \Omega \mapsto \cc$ and projection-valued functions~$Q_j : \Omega \to \cB(\Hs)$ which are analytic for each~$j$ in a finite set~$I$.
  Moreover, we assume that~$0$ is a removable singularity of all  functions~$Q_j$ and~$\lambda_j$.
\end{description}

Note that $Q_j(\alpha)$ need not be selfadjoint. Also note that $W = M(0)$ may have degenerate eigenvalues which split as~$\alpha$ moves away from~$0$. With~$\lambda_1, \dotsc, \lambda_r$ the distinct eigenvalues of~$W$ and ~$Q_1, \dotsc, Q_r$ the associated orthogonal projectors, we may write~$I = \bigcup_{i=1}^r I_i$ with $\lambda_j(0) = \lambda_i$ if and only if~$j \in I_i$.
Then, $Q_i = \sum_{j\in I_i} Q_j(0)$ and~$\{\lambda_j\}_{j \in I_i}$ is called the $\lambda_i$-group in the terminology of Kato.
The Assumption~($\tfrac 12$Sim) is more general than the following simplicity assumption, which is already rather generic from a topological point of view and sometimes easier to verify.

\begin{description}
  \item[Assumption (Sim)] Each eigenvalue~$\lambda_i$  of~$W$ is simple in the sense that the associated spectral projector~$Q_i$ is of the form~$\chi_i\chi^*_i$ for some unit vector~$\chi_i \in \Hs$.
\end{description}

One interesting advantage of Assumption~($\tfrac 12$Sim) over~(Sim) is that it can be inferred from a simple condition on~$AA^*$, thanks to the following lemma.

\begin{lemma}\label{lem:proj-implies-12Sim}
  If $\kappa^{-1}AA^*$ is an orthogonal projection for some nonzero~$\kappa \in \rr$, then Assumption \textnormal{($\tfrac 12$Sim)} is satisfied.
\end{lemma}

\begin{proof}
  Analytically extend~$M(\alpha)$ to the complex plane and consider the set~$\mathcal{C} := \{ \alpha \in \cc : |\cos (\alpha\kappa)| = 1 \}$. Then, $M(\alpha)$ is unitary for~$\alpha \in \mathcal{C}$. It can be shown that~$\mathcal{C}$ contains nontrivial curves and hence has at least one accumulation point. The lemma thus follows from Theorem~1.10 in~\cite[\S{II.1.6}]{Kat}.
\end{proof}

Now that we have clarified our assumptions, we can proceed to give the limiting behaviour of the formula for the reduced asymptotic symbol in the sample in Proposition~\ref{prop:Delta} and for the asymptotic currents in Corollary~\ref{cor:avg-curr} as~$\alpha \to 0$.

\begin{lemma}\label{lem:12Sim-imp-order2}
	If Assumption~\textnormal{($\tfrac 12$Sim)} is satisfied then for all
	$\alpha\in \Omega$, a complex neighbourhood of the origin, we have $\lambda_j(\alpha)=\lambda_j(-\alpha)$ and $Q_i(\alpha)=Q_i(-\alpha)$.
\end{lemma}

\begin{proof}
	With $N(\alpha)=W \sum_{n=0}^{+\infty} \frac{(-\alpha)^n}{(2n)!} \left(A^*A\right)^n$, we have $M(\alpha)=N(\alpha^2)$ for any $\alpha\in \Omega$, and,
	for $0<\alpha\in\Omega$,
	$N(\alpha)=\sum_{j\in I} \lambda_j(\sqrt{\alpha}) Q_j(\sqrt{\alpha})\equiv \sum_{j\in I} \mu_j(\alpha)P_j(\alpha).$
	By perturbation theory, \cite[\S{II.1}]{Kat}, the eigenvalues and eigenprojectors of $N(\alpha)$,   $\mu_j(\alpha)$ and $P_j(\alpha)$, admit analytic extensions  in $\Omega\setminus\{0\}$
	given by Laurent series in $\alpha^{1/d_j}$, $d_j\in \nn^*$.  Theorem 1.9 in \cite[\S{II.1}]{Kat}, implies $d_j=1$, since otherwise
	$\|P_j(\alpha)\|=\|Q_j(\sqrt{\alpha})\|$ diverges as $\alpha\rightarrow 0$, contradicting ($\tfrac 12$Sim.). Thus, $\mu_j$ and $P_j$ are analytic in $\Omega$ and $\lambda_j(\alpha)=\mu_j(\alpha^2)$ and $Q_j(\alpha)=P_j(\alpha^2)$ for all $\alpha\in \Omega$.
\end{proof}

\begin{theorem}\label{thm:Delta-sc}
  Suppose that Assumption~\textnormal{(Sp)} holds for all $\alpha \in \Omega \cap \rr$, that Assumptions~\textnormal{(IC)} and~\textnormal{($\tfrac 12$Sim)} hold. Then, the symbol $\Delta_\alpha^\infty$ in Proposition~\ref{prop:Delta}, which depends on the coupling strength~$\alpha$, admits an expansion
  \[
    \Delta_\alpha^\infty =  \sum_{i=1}^r \sum_{j,j' \in I_i} \frac{2}{c_j + c_{j'}}  Q_j(0) A \hat{\Xi}({ -\ii}\log\lambda_i) A^* Q_{j'}(0) + O({ \alpha^2})
  \]
  where
  \[
    c_j := \tr [Q_j(0) AA^*] > 0.
  \]
\end{theorem}

Before we proceed with the proof, let us remark that the appearance of a logarithm is due to the fact that we have defined our Fourier representation on the interval rather than on the unit circle. By periodicity of~$\hat{\Xi}$ and the fact that~$\lambda_i$ is on the unit circle, the choice of logarithm is irrelevant.

\begin{proof}[Proof of Theorem~\ref{thm:Delta-sc}]
  By Proposition~\ref{prop:Kal}, Assumption~(Sp) implies that the image of~$AA^*$ is contained in no nontrivial subspace which is stable by~$W$. Hence, $c_j := \tr [Q_j(0) AA^*] > 0$ for each~$j$. Since $M(\alpha) = W(\one - \tfrac 12 \alpha^2 AA^*) + O(\alpha^4)$, { Lemma \ref{lem:12Sim-imp-order2} and standard perturbation theory give}
  \begin{equation}
  \label{eq:exp-lambda-j}
    j \in I_i \quad \Rightarrow \quad \lambda_j(\alpha) = \lambda_i(1 - \tfrac 12 \alpha^2 c_j) + O(\alpha^4).
  \end{equation}
  \begin{description}
    \item[Claim.] The map~$\Psi$ introduced in Proposition~\ref{prop:Delta} is such that
    \[
     \alpha^2 \Psi(X) = \sum_{i=1}^r \sum_{j,j' \in I_i} \frac{2}{c_j+c_j'} Q_j(0)XQ_{j'}(0)+O(\alpha^2)
    \]
    for any linear map~$X$ on~$\Hs$.
  \end{description}
  Accepting this claim, we need only note that
  \[
    Z\SaBl = - \ii W \frac{\sin(\alpha\sqrt{AA^*})}{\sqrt{AA^*}} A = -\ii \alpha W A + O(\alpha^3)
  \]
  and the summability condition in Assumption~(IC) imply that the map~$G$ appearing in Proposition~\ref{prop:Delta} has the expansion
  \[
    G = \alpha^2 \bigg( \frac 12 W A \Xi_0A^*W^* + \sum_{k=1}^\infty W^{k+1} A \Xi_k A^* W^* \bigg) + O(\alpha^4)
  \]
  to conclude the proof.

  \begin{description}
    \item[Proof of Claim.] Inserting the spectral decomposition~\eqref{eq:spec-M} of~$M$ in Assumption~($\tfrac 12$Sim) in the definition of~$\Psi(X) := \sum_{m=0}^\infty M^m X (M^*)^m$ yields
    \begin{align*}
      \Psi(X)
        &= \sum_{j,j' \in I} \sum_{m=0}^\infty \lambda_j(\alpha)^m \overline{\lambda_{j'}(\alpha)}^{m} Q_j(\alpha) X Q_{j'}(\alpha)^* \\
        &= \sum_{j,j' \in I}  \frac{1}{1 - \lambda_j(\alpha) \overline{\lambda_{j'}(\alpha)}} Q_j(\alpha) X Q_{j'}(\alpha)^*.
    \end{align*}
    Since~$W$ is unitary we have~$Q_j(0)^* = Q_j(0)$ and $\overline{\lambda_{j'}(0)} = \lambda_{j'}(0)^{-1}$. If $\lambda_j(0) \neq \lambda_{j'}(0)$, the expansion~\eqref{eq:exp-lambda-j} gives
    \[
      \frac{1}{1 - \lambda_j(\alpha) \overline{\lambda_{j'}(\alpha)}} = \frac{1}{1 - \lambda_j(0) \lambda_{j'}(0)^{-1}} + O(\alpha^2)
    \]
    This leaves the terms for which~$\lambda_j(0) = \lambda_{j'}(0)$ (i.e.\ $j,j' \in I_i$ for some~$i$), for which we have
    \begin{align*}
      \lambda_j(\alpha) \overline{\lambda_{j'}(\alpha)}
        &= 1 - \tfrac 12 \alpha^2 (c_j + c_{j'}) + O(\alpha^4).
    \end{align*}
    by~\eqref{eq:exp-lambda-j}. Hence,
    \[
      \frac{\alpha^2}{1 - \lambda_j(\alpha) \overline{\lambda_{j'}(\alpha)}} = \frac{2}{c_j + c_{j'}} + O(\alpha^2)
    \]
    whenever $j,j' \in I_i$ for some common~$i$.
    \hfill\qed
  \end{description}
  And the Claim yields the {Theorem}.
\end{proof}

\begin{proposition}
  Suppose that Assumption~\textnormal{(Sp)} for all $\alpha \in \Omega \cap \rr$ and that Assumptions~\textnormal{(IC)}, { \textnormal{(Bl)}} and~\textnormal{($\tfrac 12$Sim)} hold. Then, with~$J_k$ as in Corollary~\ref{cor:avg-curr} depending on~$\alpha$, we have
  \begin{equation}
    J_k = \alpha^2 \tr(\Pi_k D) + O({ \alpha^4}),
  \end{equation}
  as $\Omega \ni \alpha \to 0$, where
  \[
     D = \sum_{h=1}^{r} \bigg(
      -A^* Q_h A \hat{\Xi}(-\ii\log \lambda_h)
      + \sum_{j,j' \in I_h} \frac{2}{c_j + c_{j'}} A^* Q_j(0) A \hat{\Xi}( -\ii\log \lambda_h) A^* Q_{j'}(0)A
      \bigg).
  \]
\end{proposition}

\begin{proof}
  The starting point is the expression \eqref{eq:JofX} for $J_k$.
  \begin{align*}
    Y_0&=C=\cos(\alpha\sqrt{A^*A})=I-\frac{\alpha^2}{2} A^*A+O(\alpha^4) \\
    Y_l&=Z\BlSa M^{l-1} Z\SaBl=-\alpha^2 A^* M^{l-1} W A +O({ \alpha^4}\|M^{l-1}\|),
  \end{align*}
  where $M^{l-1}=M(\alpha)^{l-1}$ is such that $\|M(\alpha)^{l-1}\|$ is uniformly bounded in $l>0$ and $\alpha\in \Omega \cap \rr$. Thus, using Equation \eqref{eq:JofX}
  we have
  \begin{equation}\label{jktheta}
  \begin{split}
    J_k=&\tr_{\Hb} \bigg[ \Pi_k \bigg(-\frac{\alpha^2}{2} (A^*A\Xi_0+\Xi_0A^*A)
    \\ & \quad\qquad \qquad
    + \sum_{l=1}^{+\infty} Y_l \Xi_l C+C\sum_{l=1}^{+\infty} \Xi_{-l} Y_l^*+\sum_{l,l'>0} Y_l \Xi_{l-l'} Y_{l'}^* \bigg) \bigg]
    +O(\alpha^4).
  \end{split}
  \end{equation}
  Let us estimate the first sum, making use of Assumptions (IC) and  ($\tfrac 12$Sim)
  \begin{align*}
  	\sum_{l=1}^{+\infty} Y_l \Xi_l C&=-\alpha^2\sum_{l=1}^{+\infty}\sum_{j\in I}  A^* Q_j(\alpha) W A \lambda_j(\alpha)^{l-1} \Xi_l +O({ \alpha^4}) \\
  		&= -\alpha^2 \sum_{j\in I} \frac{1}{\lambda_j(\alpha)} A^* Q_j(\alpha) W A \left(\sum_{l=1}^{+\infty} \lambda_j(\alpha)^{l}\Xi_l\right)+O({ \alpha^4})~.
  \end{align*}
  Thanks to $\Xi_l^*=\Xi_{-l}$, we have $\hat \Xi(\theta)=F(\theta)+F(\theta)^*=2\operatorname{Re}(F(\theta))$, where
  \[
  	F(\theta)=\frac12\Xi_0+\sum_{l\geq 1}\Exp{\ii l\theta}\Xi_l.
  \]
  Now, $\frac{1}{\lambda_j(\alpha)}Q_j(\alpha)W=Q_j(0)+O(\alpha)$, and $F$ is differentiable (since $\sum_{k\in \zz} |k| \|\Xi_k\|<+\infty$) so $F({ -\ii} \log \lambda_j(\alpha))=F({ -\ii} \log \lambda_h)+O(\alpha)$ where $h$ is such that $j\in I_h$. Taking into account the identity $\sum_{j\in I}Q_j(0)=\one$, and repeating the argument for the second sum, we get
  \begin{equation*}
  \begin{split}
    &\frac{\alpha^2}{2} ({ A^*A}\Xi_0+\Xi_0{ A^*A})-\sum_{l=1}^{+\infty} Y_l \Xi_l C-C\sum_{l=1}^{+\infty} \Xi_{-l} Y_l^*
    \\ & \qquad\qquad =\alpha^2\sum_{h=1}^r 2\operatorname{Re}(A^*Q_h A F( { -\ii}\log \lambda_h))+O({ \alpha^4}).
  \end{split}
  \end{equation*}
  The only thing left is the double sum. We consider the cases where $l=l', l<l'$ and $l>l'$ separately to write
  \begin{align*}
  	\sum_{l,l'>0} Y_l\Xi_{l-l'}Y_{l'}^*&=\sum_{l=1}^{+\infty} Y_l \Xi_0 Y_l^* +2\operatorname{Re}\left(\sum_{d>0}\sum_{l>0}   Y_{l+d} \Xi_{d} Y_l^*\right).
  \end{align*}
  Writing $M=\sum_{j\in I} \lambda_j(\alpha)Q_j(\alpha)$ and performing the summations as in the proof of Theorem \ref{thm:Delta-sc}, we obtain
  \begin{align*}
  	\sum_{l=1}^{+\infty} Y_l\Xi_0 Y_l^*&=\sum_{j,j'\in I}\frac{1}{1-\lambda_j(\alpha)\overline{\lambda_{j'}(\alpha)}} Z\BlSa Q_j(\alpha) Z\SaBl \Xi_0 Z\SaBl^* Q_{j'}(\alpha)^* Z\BlSa^*~.
  \end{align*}
  We also saw in the proof of Theorem \ref{thm:Delta-sc} that as $\alpha\rightarrow 0$
  \[
  	\frac{\alpha^2}{1-\lambda_j(\alpha)\overline{\lambda_{j'}(\alpha)}}=
  	\left\{ \begin{array}{cc}
  	\frac{2}{(c_j+c_{j'})}+O(\alpha^2) & \text{ if $\lambda_j(0)=\lambda_{j'}(0)$} \\
  	O(\alpha^2)  & \text{if $\lambda_j(0)\neq \lambda_{j'}(0)$}
  \end{array}\right.
  \]
  and since $Z\SaBl=-\ii\alpha W A+O(\alpha^3)$ and $Z\BlSa=-\ii\alpha A^*+O(\alpha^3)$ we obtain
  \[
  	\sum_{l=1}^{+\infty}Y_l \Xi_0 Y_l^*=\alpha^2\sum_{h=1}^r \sum_{j,j'\in I_h} \frac{2}{c_j+c_{j'}} A^* Q_j(0)  A \Xi_0 A^* Q_{j'}(0) A+O({ \alpha^4})~.
  \]
  Similarly, using the differentiability of $z\mapsto \sum_{d>0} z^d \Xi_d$ we have
  \[
  	\sum_{d>0}\sum_{l>0}   Y_{l+d} \Xi_d Y_l^*=\alpha^2 \sum_{h=1}^r \sum_{j,j'\in I_h}\frac{2}{c_j+c_{j'}} A^* Q_j(0) A \left(\sum_{d>0} \lambda_h^d \Xi_d\right) A^* Q_{j'}(0) A+O({ \alpha^4})~.
  \]
  Adding up all the previous estimates we get for the order-$\alpha^2$ term in parentheses in~\eqref{jktheta}
  \[
  	\sum_{h=1}^r 2\operatorname{Re}\bigg\{-A^*Q_h A F( { -\ii}\log \lambda_h)+\sum_{j,j'\in I_h} \frac{2}{c_j+c_{j'}} A^* Q_j(0)  A  F({ -\ii}\log \lambda_h) A^* Q_{j'}(0) A\bigg\}.
  \]
  Finally, the relation $[\Pi_k,F(\theta)]=0$ and the cyclicity of the trace in the definition of the current proves the proposition.
\end{proof}

For the remainder of the section, we fix
\begin{equation}\label{eq:couplingA}
  A = \sum_{k=1}^{n\Bl} \phi_k \psi_k^*
\end{equation}
for an orthonormal basis~$(\psi_k)_{k = 1}^{n\Bl}$ of~$\Hb$ and an orthonormal family~$(\phi_k)_{k=1}^{n\Bl}$ in~$\Hs$, and assume that
\[
  \hat{\T}(\theta) = \sum_{k=1}^{n\Bl} f_k(\theta) \psi_k \psi^*_k
\]
for some scalar functions~$f_k : [0,2\pi] \to [0,1]$. This corresponds to the situation from the introduction. Note that~$AA^*$ being an orthogonal projector on~$\Hs$, Lemma~\ref{lem:proj-implies-12Sim} applies.

The following proposition expresses, to leading order in the coupling parameter~$\alpha$, the currents as a sum of the contributions from channels corresponding to the eigenvalues~$\{ \lambda_i \}_{i\in I}$  associated to normalized eigenvectors~$\{ \chi_i \}_{i\in I}$ of~$W$, each expressed in terms of a simple star-shaped linear circuit.

\begin{proposition}\label{prop:circuit}
  Suppose that Assumption~\textnormal{(Sp)} holds for all~$\alpha \in \Omega \cap \rr$ and that Assumptions~\textnormal{(IC)} and~\textnormal{(Sim)} are satisfied in the setup described above. Then the symbol~$\Delta_\alpha^\infty$ admits an expansion
  \[
    \Delta_\alpha^\infty = \sum_{i=1}^r \sum_{k=1}^{n\Bl} \frac{ |\braket{\chi_i,\phi_k}|^2}{\sum_{k' = 1}^{n\Bl} |\braket{\chi_i,\phi_{k'}}|^2 } f_k({{ -\ii} \log \lambda_i})  \chi_i \chi_i^* + O({ \alpha^2})
  \]
  and the $k$-th current admits an expansion
  \[
    J_k = \alpha^2 \sum_{i \in I} J_{k,i}^{(2)} + O({ \alpha^4})
  \]
  where
  \begin{equation}\label{eq:asy-surr-2}
    J_{k,i}^{(2)} =
   \sum_{k'}\frac{|\braket{\phi_k, \chi_i}|^2|\braket{\phi_{k'}, \chi_i}|^2}{\sum_{k''=1}^{n\Bl} |\braket{\phi_{k''},\chi_i}|^2} \big(f_{k'}({ { -\ii} \log \lambda_i}) -  f_k({ { -\ii}\log \lambda_i})\big).
  \end{equation}
  Equivalently, the last equation states that the currents~$\{ J_{k,i}^{(2)} \}_{k=1}^{n_b}$ are the solutions to the classical Kirchhoff problem in Figure~\ref{fig:circuit} with voltage sources~$\{ f_k({ { -\ii}\log \lambda_i}) \}_{k=1}^{n\Bl}$ and resistors~$\{ |\braket{\phi_k,\chi_i}|^{-2} \}_{k=1}^{n\Bl}$.
\end{proposition}

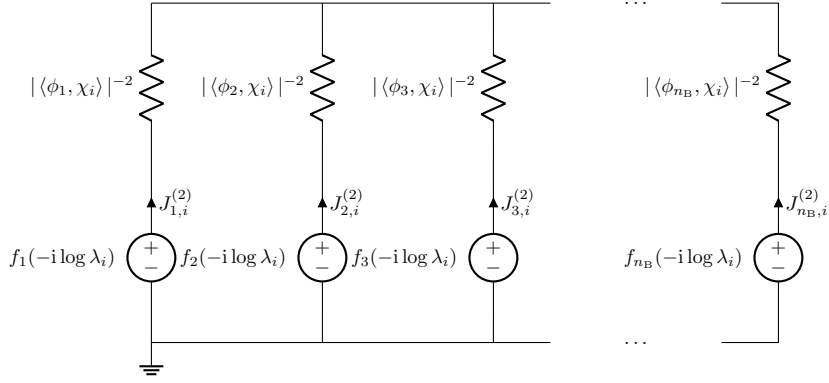
\begin{figure}
  \begin{center}
  {\small
  \begin{circuitikz}[scale = .75, transform shape]
    \draw (0,0) to[V=$f_1({ -\ii}\log \lambda_i)$,invert,i_>=$J_{1,i}^{(2)}$] (0,3) to[R=$|\braket{\phi_1,\chi_i}|^{-2}$] (0,6);
    \draw (3,0) to[V=$f_2( { -\ii}\log \lambda_i)$,invert,i_>=$J_{2,i}^{(2)}$] (3,3) to[R=$|\braket{\phi_2,\chi_i}|^{-2}$] (3,6);
    \draw (6,0) to[V=$f_3({ -\ii}\log \lambda_i)$,invert,i_>=$J_{3,i}^{(2)}$] (6,3) to[R=$|\braket{\phi_3,\chi_i}|^{-2}$] (6,6);
    \draw (10,0) -- (11,0) to[V=$f_{n\Bl}({{ -\ii}\log \lambda_i})$,invert,i_>=$J_{n\Bl,i}^{(2)}$] (11,3) to[R=$|\braket{\phi_{n\Bl},\chi_i}|^{-2}$] (11,6) -- (10,6);
    \draw (7,0) -- (0,0) node[ground] {};
    \draw (0,6) -- (7,6);
    \node at (8.5,0) {$\ldots$};
    \node at (8.5,6) {$\ldots$};
  \end{circuitikz}
  }
  \end{center}
  \caption{The currents~$(J_{k,i}^{(2)})_{k=1}^{n_b}$ in Proposition~\ref{prop:circuit} are the steady-state solutions to a linear circuit with voltage sources~$(f_k({ { -\ii}\log \lambda_i}))_{k=1}^{n_b}$ and resistors~$(|\braket{\phi_{k},\chi_i}|^{-2})_{k=1}^{n\Bl}$. Such a circuit is associated to each eigenvalue~$\lambda_i$ of~$W$.}
  \label{fig:circuit}
\end{figure}

Note that the sign of the currents is not completely determined by the properties of the initial state of the different reservoirs. While this phenomenon is {not} specific to our model, formulas such as~\eqref{eq:asy-surr-2} may allow one to explore its relation to the different phases and properties of the walk on the sample. In keeping with the illustration of the introduction, consider~$\Hb =\cc^2$, with orthonormal basis $\{\psi_1,\psi_2\}$ and note that if the functions~$f_1$ and~$f_2$ in the decomposition
 \[
   \hat T(\theta) = f_1({ \theta}) \psi_1 \psi^*_1 + f_2({ \theta}) \psi_2 \psi^*_2
 \]
 of~$T: \ell^2(\zz) \times \cc^2 \to \ell^2(\zz) \times \cc^2$ are such that neither~$f_1 \geq f_2$ or~$f_1 \leq f_2$ everywhere,
 then we can construct a unitary one-particle dynamics~$W_{\rightarrow} : \Hs \to \Hs$ in the sample and a bounded operator~$A_{\rightarrow} : \cc^2 \to \Hs$  of the form \eqref{eq:couplingA}  such that $J_1 > 0$ for all nonzero~$\alpha \in \Omega$ sufficiently small,
 as well as a unitary dynamics~$W_{\leftarrow} : \Hs \to \Hs$ in the sample and a bounded operator~$A_{\leftarrow} : \cc^2 \to \Hs$ of the form \eqref{eq:couplingA} such that $J_1 < 0$ for all nonzero $\alpha \in \Omega$ sufficiently small.
 Indeed, we can choose~$W$ to have simple eigenvalues associated to eigenvectors~$(\chi_i)_{i \in I}$ such that~$\braket{\chi_i,\phi_k} \neq 0$ for both~$k=1$ and~$k=2$.
 Then, by~\eqref{eq:asy-surr-2},
 choosing the eigenvalues in~$\{z \in \mathbf{S}^1 : f_1({ -\ii} \log z) < f_2( { -\ii}\log z)\}$ [resp. $f_1({ -\ii}\log z) > f_2({ -\ii}\log z)$] gives~$J_1 > 0$ [resp.~$J_1 < 0$] for~$\alpha$ small enough.

\begin{remark} \label{posentprodconst}
  In case $f_k(\theta)\equiv f_k$ for all $k$, Proposition \ref{prop:circuit} and Corollary \ref{sigmuj} provide the following small coupling expression of the entropy production rate
  \[
  \sigma_+=\alpha^2 \sum_{k=1}^{n\Bl}\sum_{k'=1}^{n\Bl} \mu_k  (f_{k'} - f_k)\sum_{i\in I}\frac{|\braket{\phi_k, \chi_i}|^2|\braket{\phi_{k'}, \chi_i}|^2}{\sum_{k''=1}^{n\Bl} |\braket{\phi_{k''},\chi_i}|^2}  +O({ \alpha^4})
  \]
  Setting $C^{(2)}_{k,k'}:=\sum_{i\in I}\frac{|\braket{\phi_k, \chi_i}|^2|\braket{\phi_{k'}, \chi_i}|^2}{\sum_{k''=1}^{n\Bl} |\braket{\phi_{k''},\chi_i}|^2}>0$, we have $C^{(2)}_{k,k'}=C^{(2)}_{k',k}$, $\sum_{k}C^{(2)}_{k,k'}=1$ and
  \[
  \sigma_+= \frac{\alpha^2}{2}\sum_{k\neq k'} (\mu_k-\mu_k')(f_{k'} - f_k)C^{(2)}_{k,k'}+O({ \alpha^4}).
  \]
  where the leading term is zero if and only if the summand vanishes for all pairs $k\neq k'$. Because $C^{(2)}_{k,k'}>0$ and because the function $(0,1) \ni f \mapsto \log((1-f)/f)$ defining~$\mu$ is strictly decreasing, this is in turn equivalent to $f_k = f_{k'}$ for each pair~$(k,k')$.
\end{remark}

\section{Proof of Proposition~\ref{prop:scat}}
\label{sec:scat-proofs}

The following lemma is straightforward, but we give a proof for lack of convenient reference. It can alternatively be shown to be a consequence of the Riemann--Lebesgue lemma.

\begin{lemma}
\label{lem:L1}
  {Let $\mathbf{x} = (x_n)_{n= 0}^\infty$ and~$\mathbf{y} = (y_n)_{n = 0}^\infty$ be two square-summable sequences.} Then,
  \[
    \lim_{t \to \infty}  \sum_{n=0}^{t} |x_n y_{t-n}| = 0.
  \]
\end{lemma}

\begin{proof}
  We consider $t$ even for notational simplicity. In this case,
  \begin{align*}
     \sum_{n=0}^t |x_n y_{t-n}| &\leq \sum_{d=0}^{t/2} |x_{t/2 + d} y_{t/2 - d}| + \sum_{d=1}^{t/2} |x_{t/2 - d} y_{t/2 + d}| \\
      &\leq \bigg( \sum_{d=0}^{t/2} |x_{t/2 + d}|^2\bigg)^{\frac 12} \bigg( \sum_{d=0}^{t/2} |y_{t/2 - d}|^2 \bigg)^{\frac 12}
        + \bigg( \sum_{d=1}^{t/2-1} |x_{t/2 - d}|^2\bigg)^{\frac 12} \bigg( \sum_{d=1}^{t/2} |y_{t/2 + d}|^2 \bigg)^{\frac 12} \\
      &\leq \bigg( \sum_{m = t/2}^\infty |x_{m}|^2\bigg)^{\frac 12} \|\mathbf{y}\|_{\ell^2}
        + \|\mathbf{x}\|_{\ell^2} \bigg( \sum_{m = (t/2) + 1}^\infty |y_{m}|^2 \bigg)^{\frac 12}.
  \end{align*}
	Hence, the result follows from square summability.
\end{proof}

\begin{proof}[Proof of Proposition~\ref{prop:scat}]
  The selfadjoint term being subtracted on the left-hand side of~\eqref{eq:pre-scat-t} obviously converges strongly to~$\sum_{n \geq 0} \delta_n\delta_n^* \otimes \one \oplus 0$ as~$t \to \infty$. The only explicit $t$-dependence in summands on the right-hand side of~\eqref{eq:pre-scat-t} is in the upper-right block, but the adjoint of this contribution vanishes strongly as~$t\to\infty$. To see this, combine Lemma~\ref{lem:L1} with the estimate
  \begin{align*}
    \bigg\|\sum_{m = 0}^{t-1}(\delta_{-t+m}^* \otimes   (M^*)^mZ\BlSa^* U^{m-t}) v \bigg\|
      &\leq \sum_{n = 1}^{t} \|M^{t-n}\| \| ({ \delta_{-n}^*} \otimes \one)v\|
  \end{align*}
  keeping in mind that the facts that~$v \in \ell^2(\zz) \otimes \Hb$ and that Assumption~(Sp) holds imply respectively that $\sum_{m \geq 0}\|M^m\|^2 < \infty$ and $\sum_{n \geq 0} \|(\delta_{-n}^* \otimes \one)v\|^2 < \infty$.

	Thus, in order to prove the proposition, it is sufficient to show the {strong} convergences
  \begin{align*}
    \slim\limits_{t \to \infty} \sum_{m = 0}^{t-1}\sum_{l = 1}^{t-m}\mathsf{ULB}^{-}_{m,l} &= \sum_{m \geq 0}\sum_{l \geq 1} \mathsf{ULB}^{-}_{m,l},
      & \slim\limits_{t \to \infty} \sum_{m = 0}^{t-1} \mathsf{LLB}^{-}_{m} &= \sum_{m \geq 0}\mathsf{LLB}^{-}_{m},
  \end{align*}
  where $\mathsf{ULB}^{-}_{m,l}$ and $\mathsf{LLB}^{-}_{m}$ are respectively the summands in the upper-left and lower left-block on the right-hand side of~\eqref{eq:pre-scat-t}.

  For the upper-left block, we will make use of the shorthand
  \[
    \mathbf{T}_t := \{(m,l) : 0\leq m \leq t-1; 1 \leq l \leq t-m \}.
  \]
  We want to show that the sequence of partial sums is Cauchy for the {strong} topology. {To this end, consider $v \in \ell^2(\zz)\otimes\Hb$ and natural numbers~$0 < t < u$ and note that
	\begin{align*}
		\bigg\|\sum_{(m,l) \in \mathbf{T}_{u}} \mathsf{ULB}^{-}_{m,l}v  - \sum_{(m,l) \in \mathbf{T}_{t}} \mathsf{ULB}^{-}_{m,l}v  \bigg\|^2
		&\leq
		\sum_{j} \sum_{(m,l)  \notin \mathbf{T}_{t}}  \|Y_m \|^2 \ |a_{j,-m-l}|^2
	\end{align*}
	where~$j$ ranges over the finite index set for the orthonormal basis~$\{\phi_j\}_j$ of~$\Hb$ and $(a_{j,l'})_{j,l'}$ are the  coefficients of~$v$ in the corresponding basis of $\ell^2(\zz) \otimes \Hb$. If $(m,l) \notin \mathbf{T}_{t}$, then $n := m + l \geq t$.
	Hence, the square summability of $a_{j,l'}$s and the $Y_m$s implies that
	\begin{align*}
		 \sum_{(m,l)  \notin \mathbf{T}_{t}} \|Y_m \|^2 \ |a_{j,-m-l}|^2
		&\leq  \sum_{n = t}^\infty   |a_{j,-n}|^2 \sum_{m=0}^\infty \|Y_m \|^2
	\end{align*}
	converges to 0 as $t \to \infty$ for each of the (finitely many) indices~$j$.}

  For the lower left-block, note that, for $0 < t < u$,
  \begin{align*}
    \bigg\| \sum_{m = 0}^{u-1} \mathsf{LLB}^{-}_{m}  - \sum_{m = 0}^{t-1} \mathsf{LLB}^{-}_{m} \bigg\| \leq  \sum_{m = t}^{\infty} \|\mathsf{LLB}^{-}_{m}\| ,
  \end{align*}
  with
  \[
    \|\mathsf{LLB}^{-}_{m}\| = \|\delta_{-m-1}^* \otimes  M^m Z\SaBl  U^{-m-1}\| \leq \|M^m\|.
  \]
  Again because the sequence~$(\|Y_m\|)_{m \geq 1}$ is summable, the sequence of partial sums is Cauchy in the uniform operator topology.
\end{proof}

\section{Proof of Theorem~\ref{thm:lim-EP}}
\label{sec:ep-proof}

We will make use of the following technical lemma.

\begin{lemma}
\label{eq:lem-play-ln}
  If $\epsilon \leq \T \leq (1-\epsilon)\one$ for some~$\epsilon > 0$ and $\Omega$ is a unitary operator such that~$\Omega - \one$ is trace class, then
  \[
    \tr[\Omega\T\Omega^*(\log(\Omega\T\Omega^*)-\log\T)] = \tr[(\T - \Omega\T\Omega^*) \log \T] < \infty.
  \]
\end{lemma}

\begin{proof}
  Let $\Theta$ be the trace-class operator such that~$\Omega = \one + \Theta$. Then,
  \begin{align*}
    \Omega\T\Omega^*(\log(\Omega\T\Omega^*)-\log\T)
      &= (\one + \Theta)\T\log\T(\one + \Theta^*) - (\one+\Theta)\T(\one+\Theta^*) \log \T \\
      &= \Theta\T\log\T + \T\log\T\Theta^* + \Theta\T\log\T\Theta^* \\
        &\qquad\qquad {} - \Theta\T\log\T -\T\Theta^*\log\T - \Theta\T\Theta^*\log\T.
  \end{align*}
  On the other hand,
  \begin{align*}
    (T - \Omega\T\Omega^*) \log \T
      &= (\one+\Theta^*)(\one+\Theta)\T\log\T - {(\one+\Theta)\T(\one+\Theta^*)\log\T} \\
      &= \Theta^* \T\log\T + \Theta\T\log\T + \Theta^*\Theta\T\log\T \\
        &\qquad\qquad {} - \Theta\T\log\T - \T\Theta^*\log\T - \Theta\T\Theta^*\log\T.
  \end{align*}
  All terms are trace class in each right-hand side since~$T$ and~$\log T$ are bounded. Hence, using linearity and cyclicity of the trace and the fact that~$[T,\log T] = 0$, we get
  \begin{align*}
    \tr[\Omega\T\Omega^*(\log(\Omega\T\Omega^*)-\log\T)]
      &= \tr[(\Theta^*\Theta\T - \Theta\T\Theta^*)\log\T] \\
      &= \tr[(T - \Omega\T\Omega^*)\log \T].
  \qedhere
  \end{align*}
\end{proof}

Let us recall that we are looking at the relative entropy between the quasifree states associated to the symbols~$\T\tot$ and $\T\tot(t) = \Omega(t) \T\tot\Omega^*(t)$\,---\,we have dropped some indices for readability\,---\,assuming that~$\T\tot$ has the block diagonal form
\[
  \T\tot =
  \begin{pmatrix}
    \T\En & 0 \\
    0 & \T\Sa
  \end{pmatrix}.
\]
We also decompose the unitary
\[
  \Omega(t) =
  \begin{pmatrix}
    \Omega\En(t) & \Omega\EnSa(t) \\
    \Omega\SaEn(t) & \Omega\Sa(t)
  \end{pmatrix}.
\]
{
We observe also that \eqref{eq:pre-scat-t} yields for any $t$,
\[
 \fU^t(S\otimes U \oplus \one)^{-t}= \bigg(\sum_{m \in \zz} P_m \otimes U^m \oplus \one \bigg)^* \fU_\one^t(S\otimes \one  \oplus \one)^{-t}\bigg(\sum_{m \in \zz} P_m \otimes U^m \oplus \one \bigg),
 \]
 where $\fU_\one$ is obtained from $ \fU$ by setting $U=\one$. Since the relative entropies in the definition of $\sigma(t)$ are invariant under simultaneous unitary transformation of both their arguments, we can consider $\Omega(t)$ for $U=\one$ above and consider that $ \T\En$ absorbs $U$ as described in \eqref{eq:def-Xi}.

}
It easy to see from {the results of Subsection~\ref{sec:gen-i}}
that $\Omega\EnSa(t)$, $\Omega\SaEn(t)$ and~$\Omega\Sa(t)$ have their rank bounded by~$\dim\Hs$, uniformly in~$t \geq 0$.

Let us introduce
\begin{equation}
\label{eq:fYt}
  \fY_t := \sum_{l=0}^{t-1} \sum_{m=1}^{t-l} \delta_{-m}\delta_{-m-l}^* \otimes Y_l.
\end{equation}
Then, $\operatorname{rank} \fY_t \leq t \dim \Hs$ and Proposition~\ref{prop:scat} gives $\Omega\En(t) -\one\En = -P_{[-t,-1]} \otimes \one + \fY_t$. Hence, Lemma~\ref{eq:lem-play-ln} applies and
\begin{equation}
\label{eq:ent-prod-after-trick}
  \sigma(t) = t^{-1}\tr[(\T\tot - \Omega(t)\T\tot\Omega^*(t)) \log \T\tot] + [\T\tot \mapsto \one - \T\tot],
\end{equation}
where ``${}+ [\T\tot \mapsto \one - \T\tot]$'' means to we add the same term with $\one - \T\tot$ instead of~$\T\tot$. We will show how to deal with the first of the two traces, the other one being similar. The term~$\log\T\tot$ being bounded, we consider the following representation of its multiplier
\begin{align*}
  &{ \Omega(t)\T\tot\Omega^*(t) - \T\tot} \\
  & \quad =
  \begin{pmatrix}
    \Omega\En(t)\T\En\Omega^*\En(t) + \Omega\EnSa\T\Sa\Omega^*\EnSa(t) - \T\En
      & \Omega\En(t)\T\En\Omega^*\SaEn(t) + \Omega\EnSa(t)\T\Sa\Omega^*\Sa(t) \\
    \Omega\SaEn(t)\T\En\Omega^*\En(t) + \Omega\Sa(t)\T\Sa\Omega^*\EnSa(t)
      & \Omega\SaEn(t) \T\En \Omega^*\SaEn(t) + \Omega\Sa(t)\T\Sa\Omega^*\Sa(t) - \T\Sa
  \end{pmatrix}.
\end{align*}
Note that the rank of the lower-right block is bounded by~$\dim\Hs$ and hence cannot contribute to the limit of~\eqref{eq:ent-prod-after-trick}. The same is true for each term in which~$\T\Sa$ appears. Hence, provided that the limit exists, we must have
\begin{equation}
\label{eq:ent-prod-reduced}
  \sigma^+ =  \lim_{t\to\infty}  t^{-1}\tr[(\T\En - \Omega(t)\T\En\Omega^*(t)) \log \T\En] + [\T\En \mapsto \one - \T\En].
\end{equation}
Proposition~\ref{prop:scat} yields
\begin{align}
\Omega\En(t)  \T\En \Omega\En^*(t)-\T\En &=(P^\perp_{[-t,-1]}\otimes\one) \T\En (P^\perp_{[-t,-1]}\otimes\one)-\T\En  \nonumber \\
  &  \quad +\fY_t\T\En(\one\En-P_{[-t,-1]}\otimes\one) +(\one\En-P_{[-t,-1]}\otimes\one)\T\En\fY^*_t+\fY_t\T\En\fY_t^* \nonumber,
\end{align}
where the operator on the second line has finite rank since $\fY_t$ does.
The first line of the right hand side above writes
\begin{multline}
(P^\perp_{[-t,-1]}\otimes\one) \T\En (P^\perp_{[-t,-1]}\otimes\one)-\T\En
   =(P_{[-t,-1]}\otimes\one) \T\En (P_{[-t,-1]}\otimes\one)\\
  -  \T\En (P_{[-t,-1]}\otimes\one) - (P_{[-t,-1]}\otimes\one)\T\En ,
\end{multline}
where $P_{[-t,-1]}$ has rank $t$, so that altogether, each term in this composition of $ \Omega(t)\T\En { \Omega^*(t)} - \T\En$ has finite rank of order $t$.

Let us now spell out what is left of the (first) trace in~\eqref{eq:ent-prod-reduced} dropping the tensored identities for readability:
\begin{align*}
  &\tr[\T\En P_{[-t,-1]} \log(\T\En)] + \tr[P_{[-t,-1]} \T\En P_{[-t,-1]}^\perp \log(\T\En)]
    \\ & \qquad - \tr[\fY_t \T\En P_{[-t,-1]}^\perp \log(\T\En) + \textnormal{h.c.}] - \tr[\fY_t \T\En\fY_t^* \log(\T\En)].
\end{align*}
We have used yet again cyclicity of the trace, as well as the identity
\[
  \fY_t = P_{[-t,-1]} \fY_t P_{[-t,-1]}
\]
following immediately from the definition.

{
  By invariance under translations and selfadjointness, the matrix-valued sequences, $(G^{i}_{l})_{l\in \zz}$, $i=0,1,2$, defined by
	\begin{align*}
		\braket{\phi', G^0_l \phi} &= \braket{\delta_m \otimes \phi', \T\En \log\T\En (\delta_{m+l} \otimes \phi)}, \\
		\braket{\phi', G^1_l \phi} &= \braket{\delta_m \otimes \phi', \T\En (\delta_{m+l} \otimes \phi)}, \\
		\braket{\phi', G^2_l \phi} &= \braket{\delta_m \otimes \phi', \log\T\En (\delta_{m+l} \otimes \phi)},
	\end{align*}
	do not depend on the choice of~$m$ and satisfy $(G^i_{l})^* = G^i_{-l}$. Because $\T\En\log\T\En$ is a bounded operator,
	\[
		\|\T\En\log\T\En(\delta_0\otimes\phi)\|^2 \leq \|\T\En\log\T\En\|^2 \|\phi\|^2
	\]
	is finite for all~$\phi \in \H\Bl$. Noting that
	\begin{align*}
		\sum_{j=1}^d \|\T\En\log\T\En(\delta_0\otimes \phi_j) \|^2
			&= \limsup_{n \to \infty} \sum_{j=1}^d \|(P_{[-n,n]} \otimes \one)\T\En\log\T\En(\delta_0\otimes \phi_j) \|^2 \\
			&= \limsup_{n \to \infty}  \sum_{l=-n}^n \sum_{j,j'=1}^d \langle  \T\En\log\T\En(\delta_0 \otimes \phi_j), \\
			&\qquad\qquad\qquad\qquad\qquad (\delta_{-l} \delta_{-l}^* \otimes \phi_{j'}\phi_{j'}^*) \T\En\log\T\En(\delta_0 \otimes \phi_j)\rangle \\
			&= \limsup_{n \to \infty} \sum_{l=-n}^n \tr[(G^0_l)^*G^0_l]
	\end{align*}
  for any orthonormal basis $(\phi_j)_{j=1}^d$, it follows that
  \begin{align*}
    \|\hspace{-0.2ex}|G^0|\hspace{-0.2ex}\|^2 &:=  \sum_{l\in \zz} \tr[(G^0_l)^*G^0_l] \leq d \|\T\En\log\T\En\|^2 < \infty.
  \end{align*}
	Similarly, $\|\hspace{-0.2ex}|G^1\|\hspace{-0.2ex}|, \|\hspace{-0.2ex}| G^2\|\hspace{-0.2ex}| < \infty$.
  It is then easy to show using the H\"older inequality for trace norms and the decay of the sequence $(\|Y_l\|)_{l = 1}^\infty$ that the following three bounds hold
  \begin{gather}
    \sum_{n \in \zz} |\tr[G^i_n G^j_{-n}]| \leq \|\hspace{-0.2ex}|G^i|\hspace{-0.2ex}\| \|\hspace{-0.2ex}|G^j|\hspace{-0.2ex}\| < \infty, \label{eq:absconv1} \\
    \sum_{l = 0}^\infty \sum_{n \in \zz}  |\tr[Y_l G^i_{l+n} G^j_{-n}]| \leq \sum_{l = 0}^\infty \|Y_l\| \|\hspace{-0.2ex}|G^i|\hspace{-0.2ex}\| \|\hspace{-0.2ex}|G^j|\hspace{-0.2ex}\| < \infty, \label{eq:absconv2}\\
    \sum_{l, l' = 0}^\infty \sum_{n \in \zz}^\infty |\tr[Y_l G^i_{n-l'+l} Y_{l'} G^j_{-n}]| \leq  \sum_{l,l'=0}^\infty \|Y_l\| \|Y_{l'}\| \|\hspace{-0.2ex}|G^i|\hspace{-0.2ex}\| \|\hspace{-0.2ex}|G^j|\hspace{-0.2ex}\| < \infty. \label{eq:absconv3}
  \end{gather}
	The Fourier transforms are defined accordingly,
  \[
    \hat{G}^i(\theta) := \sum_{l \in \zz} \Exp{\ii l \theta} G^i_l,
  \]
  and satisfy
  \begin{equation}
  \label{eq:prod-prop-G}
    \hat{G}^0 = \hat{G}^1 \hat{G}^2.
  \end{equation}
}

\begin{lemma}
\label{lem:conv-warm-up}
  Under the  hypotheses of Theorem \ref{thm:lim-EP},
  \[
     t^{-1} \tr[P_{[-t,-1]} \T\En \log(\T\En)] = \int_0^{2\pi} \tr[\hat{\Xi}(\theta) \log\hat{\Xi}(\theta) ] \, \frac{\d\theta}{2\pi}
  \]
  for all~$t > 0$, with~$\hat{\Xi}$ the Fourier of transform of~$\Xi$ according to the conventions of Section~\ref{sec:Fourier}.
\end{lemma}

\begin{proof}
  {
  On one hand, we have
  \begin{align}
  \label{eq:G0-avg}
    \int_0^{2\pi} \tr[\hat{\Xi}(\theta)\log\hat{\Xi}(\theta)] \frac{\d\theta}{2\pi}
      &= \int_0^{2\pi} \tr[\hat{G}^0(\theta)] \frac{\d\theta}{2\pi} = \tr[G^0_0].
  \end{align}
  On the other hand, we have
  \begin{align*}
    \tr[P_{[-t,-1]} \T\En \log(\T\En)] &= \sum_{n=1}^t \sum_{l,m\in \zz} \delta_{-n}^* \delta_m \delta_{m+l}^* \delta_{-n} \tr[G^0_l] \\
    &= \sum_{n=1}^t \sum_{l\in \zz} \delta_{-n+l}^* \delta_{-n} \tr[G^0_l]
    = t \tr[G_0^0],
  \end{align*}
  hence the equality.
 }

\end{proof}

\begin{lemma}
  Under the ongoing hypotheses,
  \[
    \lim_{t\rightarrow +\infty} t^{-1} \tr[ P_{[-t,-1]} \T\En P_{[-t,-1]}^\perp \log(\T\En)] = 0.
  \]
\end{lemma}

\begin{proof}
  In view of Lemma~\ref{lem:conv-warm-up} and the definition of~$P_{[-t,-1]}^\perp$, the claim will be proved if we can show that
  \[
     \lim_{t\rightarrow +\infty} t^{-1} \tr[P_{[-t,-1]} \T\En P_{[-t,-1]} \log(\T\En)] = \int_0^{2\pi} \tr[\hat{\Xi}(\theta) \log\hat{\Xi}(\theta) ] \, \frac{\d\theta}{2\pi}.
  \]
  {
    We have
    \begin{align*}
      &\tr[P_{[-t,-1]} \T\En P_{[-t,-1]} \log(\T\En)] \\
      &\qquad  = \sum_{n,n'=1}^t \sum_{l,l',m,m'\in \zz} \delta_{-n}^* \delta_m \delta_{m+l}^* \delta_{-n'} \delta_{-n'}^* \delta_{m'} \delta_{m'+l'}^* \delta_{-n} \tr[G^1_l G^2_{l'}] \\
       &\qquad  = \sum_{n,n'=1}^t \sum_{l,l'\in \zz} \delta_{-n+l}^* \delta_{-n'} \delta_{-n'+l'}^* \delta_{-n} \tr[G^1_l G^2_{l'}] \\
      &\qquad  = \sum_{l''=-t+1}^{t-1} \min\{t-l'',t+l''\} \tr[G^1_{l''} G^2_{-l''}].
    \end{align*}
    Here, $\min\{t-l'',t+l''\}$ is the number of pairs $(n,n')$ between~$1$ and~$t$ satisfying $n-n' = l''$.
    Thus, in view of~\eqref{eq:prod-prop-G} and~\eqref{eq:G0-avg}, the rest
    \begin{align*}
      r_t:= \left|t^{-1} \tr[P_{[-t,-1]} \T\En P_{[-t,-1]} \log(\T\En)]-\int_0^{2\pi} \tr[\hat{\Xi}(\theta) \log\hat{\Xi}(\theta) ] \, \frac{\d\theta}{2\pi}\right|
    \end{align*}
    satisfies
    \begin{align*}
    r_t&\leq \sum_{l\in \zz} \min\left\{1,\tfrac{|l|}{t} \right\} \left|\tr (G^1_{l} G^2_{-l})\right|.
    \end{align*}
    Given the absolute convergence expressed in equation \eqref{eq:absconv1} it is easily deduced that $r_t\rightarrow 0$ as $t \to \infty$.
  }
\end{proof}

\begin{lemma}
  Under the ongoing hypotheses,
  \[
    \lim_{t \to \infty} t^{-1} \tr[\fY_t \T\En P_{[-t,-1]}^\perp \log(\T\En)] = 0.
  \]
\end{lemma}

\begin{proof}
  We have
  \begin{align*}
    \tr[\fY_t \T\En \log(\T\En)]&=\sum_{n\in\zz}\sum_{l=0}^{t-1}\sum_{m=1}^{t-l} \sum_{m',l'\in \zz}
    \delta_{n}^* \delta_{-m} \delta^*_{-m-l} \delta_{m'} \delta^*_{m'+l'} \delta_n \tr[Y_l G^0_{l'}] \\
    &=\sum_{l=0}^{t-1} (t-l) \tr[Y_l G^0_l],
  \end{align*}
  while (similarly)
  \begin{align*}
    &\tr[\fY_t \T\En P_{[-t, -1]} \log(\T\En)] \\
    &\qquad =\sum_{l=0}^{t-1}\sum_{m=1}^{t-l} \sum_{m'=-t}^{-1}
    \tr[Y_l G^1_{(m'+m)+l}G^2_{-(m'+m)}] \\
    &\qquad =\sum_{l=0}^{t-1} \sum_{n=1-t}^{t-l-1} \min\{t-l-n,t-l,t+n\} \tr[Y_l G^1_{n+l}G^2_{-n}].
  \end{align*}
  Here, $\min\{t-l-n,t-l,t+n\}$ is the cardinality of the set of pairs $(m,m')$ within the prescribed intervals such that $m'+m=n$.
  In view of~\eqref{eq:prod-prop-G} and~\eqref{eq:absconv2}, both $t^{-1}\tr[\fY_t \T\En \log(\T\En)]$ and $t^{-1}\tr[\fY_t \T\En P_{[-t, -1]}\log(\T\En)]$ converge to the (absolutely convergent) sum
  \[
    \sum_{l = 0}^\infty \sum_{n\in \zz} \tr[Y_l G^1_{n+l}G^2_{-n}].
  \]
  The lemma follows by taking the difference.
\end{proof}

\begin{lemma}
  Under the ongoing hypotheses,
  \[
    \lim_{t \to \infty} t^{-1} \tr[\fY_t \T\En \fY_t^* \log(\T\En)] = \int_0^{2\pi} \tr[\hat{\fY}(\theta) \hat{\Xi}(\theta) \hat{\fY}^* (\theta)\log(\hat{\Xi}(\theta))]\, \frac{\d\theta}{2\pi}.
  \]
\end{lemma}

\begin{proof}
  We have
  \begin{align*}
    \int_0^{2\pi} \tr[\hat{\fY}(\theta) \hat{\Xi}(\theta) \hat{\fY}^* (\theta)\log(\hat{\Xi}(\theta))]\, \frac{\d\theta}{2\pi}
    &=\int_0^{2\pi} \sum_{l, l'\geq 0} \sum_{m,m'\in \zz} \tr[Y_l G^1_{m}Y_{l'}^*G^2_{m'}] \Exp{\ii(m+m'-l+l')\theta} \frac{\d\theta}{2\pi} \\
    &=\sum_{l,l'\geq 0} \sum_{m\in \zz} \tr[Y_l G^1_{m} Y_{l'} G^2_{l-l'-m}].
  \end{align*}
	{
  On the other hand,
  \begin{align*}
    \tr[\fY_t \T\En \fY_t^* \log(\T\En)]
      &=  \sum_{l,l'=0}^{t-1} \sum_{m=1}^{t-l} \sum_{m'=1}^{t-l'}
       \tr[Y_l G^1_{m-m'+l-l'} Y_{l'}^* G^2_{m'-m}]\\
      &=\sum_{l,l'=0}^{t-1}\sum_{n = -t +l'+1}^{t-l-1} \min\{t-l-n,t-l,t-l',t-l'+n\}
        \\ & \qquad\qquad\qquad\qquad \tr[Y_l G^1_{n+l-l'}Y_{l'}^*G^2_{-n}]
  \end{align*}
  where we performed the change of variables $n=m-m'$ and $\min\{t-l-n,t-l,t-l',t-l'+n\}$ is the cardinality of the set of pairs $(m,m')$ within the prescribed intervals such that $m-m'=n$. Thus the rest in the statement of the lemma is
  \begin{align*}
    r_t \leq \sum_{l,l'=0}^{\infty}\sum_{ n \in \zz} c_{l,l',n,t}  \big|\tr[Y_l G^1_{n+l-l'}Y_{l'}^*G^2_{-n}]\big|
  \end{align*}
  where $1 \geq |c_{l,l',n,t}| \to 0$ as $t \to \infty$ for fixed~$l,l'$ and~$n$. Hence, absolute convergence in equation \eqref{eq:absconv3} yields that the rest~$r_t \to 0$ as $t \to \infty$, and hence the lemma.
  }
\end{proof}

\bibliographystyle{alpha-custom}
\bibliography{walkers-out-of-eq}

\end{document}